\documentclass{article}
\usepackage{amsfonts}
\usepackage{amsmath}

\newtheorem{theorem}{Theorem}

\newtheorem{assumption}{Assumption}

\newtheorem{definition}[theorem]{Definition}

\newtheorem{proposition}[theorem]{Proposition}

\newenvironment{proof}[1][Proof]{\noindent \textbf{#1.} }{\  \rule{0.5em}{0.5em}}

\newcommand{\bq}{\begin{equation*}}
\newcommand{\eq}{\end{equation*}}
\newcommand{\bqn}{\begin{equation}}
\newcommand{\eqn}{\end{equation}}
\newcommand{\bqq}{\begin{eqnarray*}}
\newcommand{\eqq}{\end{eqnarray*}}
\newcommand{\bqqn}{\begin{eqnarray}}
\newcommand{\eqqn}{\end{eqnarray}}

\input{tcilatex}
\begin{document}

\title{Representation of homothetic forward performance processes in
stochastic factor models via ergodic and infinite horizon BSDE\thanks{%
The authors would also like to thank the Editor, the Associate
Editor and two referees for their valuable suggestions which led to
a much improved version of the paper. The authors would like to
thank the Oxford-Man Institute for its hospitality and support of
their visits there during which most of this work was produced. The
work was presented at seminars and workshops at King's College
London, Princeton U., Shandong U., U. of Freiburg and UCSB. The
authors would like to thank the participants for their helpful
comments.}}
\author{Gechun Liang\thanks{%
Dept.\ of Mathematics, King's College, London and the Oxford-Man Institute,
University of Oxford; email: \texttt{gechun.liang@kcl.ac.uk}} \and Thaleia
Zariphopoulou\thanks{%
Depts.\ of Mathematics and IROM, The University of Texas at Austin, and the
Oxford-Man Institute, University of Oxford; email: \texttt{%
zariphop@math.utexas.edu}.} }
\date{First version: April 2014; This version: \today }
\maketitle

\begin{abstract}
In an incomplete market, with incompleteness stemming from stochastic
factors imperfectly correlated with the underlying stocks, we derive
representations of homothetic (power, exponential and logarithmic) forward
performance processes in factor-form using ergodic BSDE. We also develop a
connection between the forward processes and infinite horizon BSDE, and,
moreover, with risk-sensitive optimization. In addition, we develop a
connection, for large time horizons, with a family of classical homothetic
value function processes with random endowments.
\end{abstract}

\section{Introduction}

This paper contributes to the study of homothetic forward performance
processes, namely, of power, exponential and logarithmic type, in a
stochastic factor market model. Stochastic factors are frequently used to
model the predictability of stock returns, stochastic volatility and
stochastic interest rates (for an overview of the literature, we refer the
reader to the review paper \cite{Z-RADON}). Forward performance processes
were introduced and developed in \cite{MZ-first}, \cite{MZ0} and \cite%
{MZ-Carmona} (see, also, \cite{MZ-Kurtz}, \cite{MZ1} and \cite{MZ2}). They
complement the classical expected utility paradigm in which the utility is a
deterministic function chosen at a single point in time (terminal horizon).
The value function process is, in turn, constructed backwards in time, as
the Dynamic Programming Principle yields. As a result, there is limited
flexibility to incorporate updating of risk preferences, rolling horizons,
learning and other realistic \textquotedblleft forward in nature" features,
if one requires that time-consistency is being preserved at all times.
Forward investment performance criteria alleviate some of these shortcomings
and offer the construction of a genuinely dynamic mechanism for evaluating
the performance of investment strategies as the market evolves across
(arbitrary) trading horizons.

In \cite{MZ3} a stochastic PDE (cf. (\ref{SPDE}) herein) was proposed for
the characterization of forward performance processes in a market with It%
\^{o}-diffusion price processes. It may be viewed as the forward analogue of
the finite-dimensional classical Hamilton-Jacobi-Bellman (HJB) equation that
arises in Markovian models of optimal portfolio choice. Like the HJB\
equation, the forward SPDE is fully nonlinear and possibly degenerate. In
addition, however, it is ill-posed and its volatility coefficient is an
input that the investor chooses while, in the classical case, the
corresponding volatility is uniquely determined from the It\^{o}
decomposition of the value function process. These features result in
significant technical difficulties and, as a result, the use of the forward
SPDE\ for general It\^{o}-diffusion market dynamics has been limited.
Results for time-monotone processes (zero forward volatility) can be found
in \cite{MZ2},\textbf{\ }and a connection between the forward performance
process and optimal portfolios has been explored in \cite{ElKaroui} (see,
also \cite{ElKaroui2}). In semi-martingale markets, an axiomatic
construction for exponential preferences can be found in \cite{Gordan}.

When the market coefficients depend explicitly on stochastic
factors, as herein, there is more structure that can be explored by
seeking performance criteria represented as \textit{deterministic
functions} of these factors. As it was first noted in \cite{MZ3},
the SPDE\ reduces to a finite-dimensional HJB equation (see equation
(51) therein) that these functions are expected to satisfy. Still,
however, this HJB\ equation remains ill-posed and how to solve it is
an open problem.

For a single stochastic factor, two cases have been so far analyzed,
specifically, for the power and exponential cases. The power case was
treated in \cite{NZ} where the homotheticity reduces the forward HJB\ to a
semilinear PDE which is, in turn, linearized using a distortion
transformation. One then obtains a one-dim. ill-posed linear equation with
state dependent coefficients, which is solved using an extension of Widder's
theorem. The exponential case was studied in \cite{MZ-Kurtz} (see, also,
\cite{MZ0} and \cite{LSZ}) in the context of\textit{\ }forward exponential
indifference prices.

Multi-factor modeling of forward performance processes is considered in \cite%
{NT}, where the complete market setting is analyzed in detail. Because of
market completeness, the Legendre-Fenchel transformation linearizes the
forward SPDE,\ and a multi-dim. ill-posed linear equation with space/time
dependent coefficients arises. Its solutions are, in turn, characterized via
an extension of Widder's theorem developed by the authors. More recently,
multi-factors of different (slow and fast) scales in incomplete markets were
studied in \cite{SSZ}, and asymptotic expansions were derived for the
limiting regimes. Therein, the leading order terms are expressed as
time-monotone forward performances with appropriate stochastic
time-rescaling, resulting from averaging phenomena. The first order terms
reflect compiled changes in the investor's preferences based on market
changes and her past performance.

Herein, we initiate a study to generalize the existing results on forward
processes in factor-form allowing for market incompleteness, multi-stocks
and multi-stochastic factors. We first focus on homothetic processes (power,
exponential and logarithmic), for these are also the popular choices of risk
preferences in the classical setting.

For such cases, the homotheticity reduces the forward SPDE to an
ill-posed multi-dimensional semilinear PDE (cf. (\ref{f-eqn}),
(\ref{exponential-equation})), which however \textit{cannot }be
linearized. To our knowledge, no results exist to date for such
ill-posed equations. The main contribution herein is that we bypass
the difficulties generated by the ill-posedness by constructing
factor-form forward processes \textit{directly} from
\textit{Markovian }solutions of a family of \textit{ergodic BSDE}.
While the form of their driver is suggested by the operator
appearing in the ill-posed PDE, we use exclusively results from
ergodic equations to construct the forward solutions and not from
(forward) stochastic optimization. As a by-product, we use these
findings to construct a smooth solution to the ill-posed
multi-dimensional semilinear PDE. To our knowledge, this approach is
new. It is quite direct and requires mild assumptions on the
dynamics of the factors, essentially the ergodicity condition
(\ref{exponentialcondition}).

The second contribution is that we provide a connection with risk-sensitive
optimization and the constant appearing in the solution of the ergodic BSDE.
Thus, we provide a new interpretation, in the context of forward
optimization, of the classical results of \cite{Bielecki}, \cite{FS2} and
\cite{FS3} on the optimal growth rate of long-term utility maximization
problems.

In a different direction, we develop a connection of the homothetic forward
processes with \textit{infinite horizon BSDE}. Our contribution is
threefold. Firstly, we establish that the solutions of the latter are
themselves homothetic forward processes, albeit not Markovian. Secondly, we
show that as the parameter $\rho ,$ that appears naturally in these BSDE,
converges to zero, the relevant solutions will converge to their Markovian
ergodic counterparts. Thirdly, we use these infinite horizon BSDE\ to
establish a connection among the homothetic forward processes we construct
and classical analogues, specifically, finite-horizon value function
processes with an appropriately chosen terminal endowment. We show that
these value functions converge to the homothetic processes as the trading
horizon tends to infinity.

In the finite horizon setting,\textbf{\ }(quadratic)\textbf{\ }BSDE\ were
first studied in \cite{Kobylanski} and have been subsequently analyzed by a
number of authors. They constitute one of the most active areas of research
in financial mathematics, for they offer direct applications to risk
measures (\cite{Barrieu}), indifference prices (\cite%
{Imkeller1,Henderson3,MS}), and value functions for homothetic utilities (%
\cite{HU0}). Several extensions to the latter line of applications include,
among others, \cite{Morlais} and \cite{Becherer} where the results of \cite%
{HU0} were, respectively, generalized to a continuous martingale setting and
to jump-diffusions. We note that in the traditional framework, prices,
portfolios, risk measures and value functions are intrinsically constructed
\textquotedblleft backwards" in time and, thus, BSDE\ offer the ideal tool
for their analysis.

Despite the popularity of (quadratic) BSDE in the finite horizon setting,
neither their ergodic or infinite horizon counterparts have received much
attention to date. In an infinite dimensional setting, an ergodic Lipschitz
BSDE was introduced in \cite{HU2} for the solution of an ergodic stochastic
control problem; see also \cite{Sam, HU1, Richou}, and more recently \cite%
{Pham} and \cite{Hu11} for various extensions. The infinite horizon
quadratic BSDE was first solved in \cite{Briand0} by combining the
techniques used in \cite{Briand} and \cite{Kobylanski}.

To our knowledge, both types of ergodic and infinite horizon equations have
been so far motivated mainly from theoretical interest. Our results show,
however, that both types of equations are natural candidates for the
characterization of forward performance processes and their associated
optimal portfolios and wealths. It is worth mentioning that both the ergodic
and infinite horizon BSDE we consider actually turn out to be Lipschitz,
since one can show that the parts corresponding to the relevant processes $Z$
are bounded. In other words, the quadratic growth, which is the standard
assumption in the finite setting, does not play a crucial role. Indeed, as
we show in the Appendix, the existing results from the ergodic Lipschitz
BSDE \cite{HU2} and the infinite horizon Lipschitz BSDE \cite{Briand} can be
readily adapted to solve the forward equations at hand.

We conclude by mentioning that while we focus on forward processes
in factor-form, most of the results also apply for non-Markovian
forward processes (e.g. results in section 3.1.3). Furthermore, we
stress that a measure transformation (see, examples in 3.1.3) might
indicate that one can construct new homothetic forward processes
directly from the ones with zero volatility, thus making the results
herein redundant. However, this is \textit{not} the case. From the
one hand, changing measure corresponds to changing the risk premia,
which essentially amounts to changing the original market model.
Therefore, one does not produce any genuinely new forward processes
within the original market. More importantly, zero volatility
forward processes are decreasing in time and path-dependent with
regards to the stochastic factors. It is \textit{not} possible to
produce from them their Markovian counterparts using a measure
change transformation.

The paper is organized as follows. In section \ref{SectionModel}, we
introduce the market model, and review the notion of forward performance
process and the forward SPDE. In sections \ref{SectionPowerUtility}, \ref%
{SectionExpUtility} and \ref{SectionLogUtility}, we construct the
corresponding forward performance processes in factor-form, and the
associated optimal portfolios and wealth processes. In each section, we also
present the connection with an ill-posed semilinear PDE, as well as with the
(non-Markovian) solutions of the related infinite horizon BSDE and with
finite horizon counterparts. For the reader's convenience, we present the
technical background results on the ergodic and infinite horizon BSDE in the
Appendix.


\section{The stochastic factor model and its forward performance process}

\label{SectionModel}

The market consists of a riskless bond and $n$ stocks. The bond is taken to
be the numeraire and the individual (discounted by the bond) stock prices $%
S_{t}^{i},$ $t\geq 0,$ solve, for $i=1,...,n,$%
\begin{equation}
\frac{dS_{t}^{i}}{S_{t}^{i}}=b^{i}(V_{t})dt+\sum_{j=1}^{d}\sigma
^{ij}(V_{t})dW_{t}^{j},  \label{stock-SDE}
\end{equation}%
with $S_{0}^{i}>0.$ The process $W=(W^{1},\cdots ,W^{d})^{T}$ is a standard $%
d$-dimensional Brownian motion on a filtered probability space $(\Omega ,%
\mathcal{F},\mathbb{F}=\{ \mathcal{F}_{t}\}_{t\geq 0},\mathbb{P})$ satisfying
the usual conditions . The superscript $T$ denotes the matrix transpose.

The $d$-dimensional process $V=(V^{1},\cdots ,V^{d})$ models the stochastic
factors affecting the dynamics of stock prices, and its components are
assumed to solve, for $i=1,...,d,$%
\begin{equation}
dV_{t}^{i}=\eta ^{i}(V_{t})dt+\sum_{j=1}^{d}\kappa ^{ij}dW_{t}^{j},
\label{factor-SDE}
\end{equation}%
with $V_{0}^{i}\in \mathbb{R}$.

We introduce the following model assumptions.

\begin{assumption}
\label{Assumption1}

i) The market coefficients $b(v)=(b^{i}(v))$ and $\sigma (v)=(\sigma
^{ij}(v)),$ $1\leq i\leq n,1\leq j\leq d,$ $v\in \mathbb{R}^{d},$ are
uniformly bounded and the volatility matrix $\sigma (v)$ has full row rank $%
n $.

ii) The market price of risk vector $\theta (v),$ $v\in \mathbb{R}^{d},$
defined as the solution to the equation $\sigma (v)\theta (v)=b(v)$ and
given by $\theta (v)=\sigma (v)^{T}[\sigma (v)\sigma (v)^{T}]^{-1}b(v),$ is
uniformly bounded and Lipschitz continuous.
\end{assumption}

\begin{assumption}
\label{Assumption3}

The drift coefficients of the stochastic factors satisfy the dissipative
condition
\begin{equation}
(\eta (v)-\eta (\bar{v}))^{T}(v-\bar{v})\leq -C_{\eta }|v-\bar{v}|^{2},
\label{dissipative}
\end{equation}%
for any $v,\bar{v}\in \mathbb{R}^{d}$ and a constant $C_{\eta }$ large
enough. The volatility matrix $\kappa =(\kappa ^{ij})$, $1\leq i,j\leq d$,
is a constant matrix with $\kappa \kappa ^{T}$ positive definite and
normalized to $|\kappa |=1$.
\end{assumption}

The \textquotedblleft large enough"\ property of the above constant $C_{\eta
}$ will be refined later on when we introduce another auxiliary constant $%
C_{v}$ (cf. (\ref{driver0}) and example in Section \ref{SectionExample})
related to the drivers of the upcoming BSDE.


The dissipative condition (\ref{dissipative}) implies that the stochastic
factor process $V$ admits a unique invariant measure, and it is, thus,
ergodic. Indeed, a direct application of Gronwall's inequality yields that $%
V $ satisfy, for any $v,\bar{v}\in \mathbb{R}^{d},$ the \textit{exponential}
\textit{ergodicity} condition
\begin{equation}
|V_{t}^{v}-V_{t}^{\bar{v}}|^{2}\leq e^{-2C_{\eta }t}|v-\bar{v}|^{2},
\label{exponentialcondition}
\end{equation}%
where the superscript $v$ denotes the dependence on the initial condition.

Inequality (\ref{exponentialcondition})\textbf{\emph{\ }}states that any two
distinct\textbf{\emph{\ }}paths of the process $V$ will converge to each
other exponentially fast. We note that (\ref{exponentialcondition}) is the
only condition needed to be satisfied by the stochastic factors. Any
diffusion process satisfying inequality (\ref{exponentialcondition}) may
serve as a stochastic factor vector.

Next, we consider an investor who starts at time $t=0$ with initial
endowment $x$ and trades among the $\left( n+1\right) $ assets. We denote by
$\tilde{\pi}=(\tilde{\pi}^{1},\cdots ,\tilde{\pi}^{n})^{T}$ the proportions
of her total (discounted by the bond) wealth in the individual stock
accounts. Assuming that the standard self-financing condition holds and
using (\ref{stock-SDE}), we deduce that her (discounted by the bond) wealth
process solves
\begin{equation*}
dX_{t}^{\pi }=\sum \limits_{i=1}^{n}\tilde{\pi}_{t}^{i}X_{t}^{\pi }\frac{%
dS_{t}^{i}}{S_{t}^{i}}=X_{t}^{\pi }\tilde{\pi}_{t}^{T}\left(
b(V_{t})dt+\sigma (V_{t})dW_{t}\right) ,
\end{equation*}%
with $X_{0}=x\in \mathbb{D},$ where the set $\mathbb{D}\subseteq \mathbb{R}$
denotes the wealth admissibility domain.

For mere convenience, we will be working throughout with the trading
strategies \textit{rescaled by the volatility}, namely,
\begin{equation}
\pi _{t}^{T}=\tilde{\pi}_{t}^{T}\sigma (V_{t}).  \label{policy-normalized}
\end{equation}%
Then, the wealth process solves
\begin{equation}
dX_{t}^{\pi }=X_{t}^{\pi }\pi _{t}^{T}(\theta (V_{t})dt+dW_{t}).
\label{wealth-process}
\end{equation}%
For any $t\geq 0$, we denote by $\mathcal{A}_{{\left[ 0,t\right] }}$ the set
of admissible strategies in the trading interval $[0,t],$ given by
\begin{equation}
\mathcal{A}_{{\left[ 0,t\right] }}=\{(\pi _{u})_{u\in \lbrack 0,t]}:\pi \in
\mathcal{L}_{BMO}^{2}[0,t],\  \pi _{u}\in \Pi \text{ and }X_{u}^{\pi }\in
\mathbb{D},\text{ }u\in \left[ 0,t\right] \}.  \label{admissibility}
\end{equation}%
The set $\Pi \subseteq \mathbb{R}^{d}$ is closed and convex, and the space $%
\mathcal{L}_{BMO}^{2}[0,t]$ defined as%
\begin{equation*}
\mathcal{L}_{BMO}^{2}[0,t]=\  \left \{ (\pi _{u})_{u\in \lbrack 0,t]}:\pi \
\text{is}\  \mathbb{F}\text{-progressively\ measurable and}\right.
\end{equation*}%
\begin{equation*}
\left. ess\sup_{\tau }E_{\mathbb{P}}\left( \left. \int_{\tau }^{t}|\pi
_{u}|^{2}du\right \vert \mathcal{F}_{\tau }\right) <\infty ,\  \text{for any}%
\  \mathbb{F}\text{-stopping time}\  \tau \in \lbrack 0,t]\right \} .
\end{equation*}%
The above integrability condition is also called the \textit{BMO-condition},
since for any $\pi \in \mathcal{L}_{BMO}^{2}[0,t]$,
\begin{equation*}
ess\sup_{\tau \in \lbrack 0,t]}E_{\mathbb{P}}\left( \left. \int_{\tau
}^{t}\pi _{u}^{T}dW_{u}\right \vert \mathcal{F}_{\tau }\right)
^{2}=ess\sup_{\tau \in \lbrack 0,t]}E\left( \left. \int_{\tau }^{t}|\pi
_{u}|^{2}du\right \vert \mathcal{F}_{\tau }\right) <\infty ,
\end{equation*}%
and, hence, the stochastic integral $\int_{0}^{s}\pi _{u}^{T}dW_{u},$ $s\in %
\left[ 0,t\right] ,$ is a \textit{BMO}-martingale.

In turn, we define the set of \textit{admissible strategies} for all $t\geq
0 $ as $\mathcal{A}:=\cup _{t\geq 0}\mathcal{A}_{[0,t]}$.

Next, we review the notion of forward performance process, introduced and
developed in \cite{MZ0}-\cite{MZ3}. Variations and relaxations of the
original definition can be also found in \cite{Berrier}, \cite{ElKaroui},
\cite{Henderson-Hobson} and \cite{NT}.

\begin{definition}
\label{def} A process $U\left( x,t\right) ,$ $\left( x,t\right) \in \mathbb{%
D\times }\left[ 0,\infty \right) $ is a forward performance process if

i) for each $x\in \mathbb{D},$ $U\left( x,t\right) $ is $\mathbb{F}$%
-progressively measurable;

ii)\ for each $t\geq 0$, the mapping $x\mapsto U(x,t)$ is strictly
increasing and strictly concave;

iii)\ for any $\pi \in \mathcal{A}$ and $0\leq t\leq s,$
\begin{equation}
E_{\mathbb{P}}\left( U(X_{s}^{\pi },s)|\mathcal{F}_{t}\right) \leq U\left(
X_{t}^{\pi },t\right) ,  \label{supermartingale}
\end{equation}%
and there exists an optimal portfolio $\pi ^{\ast }\in \mathcal{A}$ such
that, for $0\leq t\leq s,$
\begin{equation}
E_{\mathbb{P}}\left( U_{s}(X_{s}^{\pi ^{\ast }},s)|\mathcal{F}_{t}\right)
=U\left( X_{t}^{\pi ^{\ast }},t\right) .  \label{martingale}
\end{equation}
\end{definition}

As mentioned earlier, it was shown in \cite{MZ3} that $U\left( x,t\right) $
is associated with an ill-posed fully nonlinear SPDE, which plays the role
of the Hamilton-Jacobi-Bellman equation in the classical finite-dimensional
setting. Formally, this forward SPDE is derived by first assuming that $%
U\left( x,t\right) $ admits the It\^{o} decomposition
\begin{equation*}
dU(x,t)=b(x,t)dt+a(x,t)^{T}dW_{t},
\end{equation*}%
for some $\mathbb{F}$-progressively measurable processes $a(x,t)$ and $%
b(x,t),$ and that all involved quantities have enough regularity so that the
It\^{o}-Ventzell formula can be applied to $U(X_{s}^{\pi },s),$ for all
admissible $\pi .$ The requirements (\ref{supermartingale}) and (\ref%
{martingale}) then yield that, for a chosen volatility process $a\left(
x,t\right) ,$ the drift $b\left( x,t\right) $ must have a specific form.

In the setting herein, the forward performance SPDE\ takes the form
\begin{align}
dU(x,t)=& \left( -\frac{1}{2}x^{2}U_{xx}(x,t)dist^{2}\left( \Pi ,-\frac{%
\theta (V_{t})U_{x}(x,t)+a_{x}\left( x,t\right) }{xU_{xx}(x)}\right) \right.
\notag \\
& \left. +\frac{1}{2}\frac{|\theta (V_{t})U_{x}(x,t)+a_{x}\left( x,t\right)
|^{2}}{U_{xx}(x,t)}\right) dt+a(x,t)^{T}dW_{t},  \label{SPDE}
\end{align}%
where $dist\left( \Pi ,x\right) $ represents the distance function from $%
x\in \mathbb{R}^{d}$ to $\Pi $. Furthermore, if a strong solution to (\ref%
{wealth-process}) exists, say $X_{t}^{\pi ^{\ast }},$ when the feedback
policy
\begin{equation}
\pi _{t}^{\ast }=Proj_{\Pi }\left( -\frac{\theta (V_{t})U_{x}(X_{t}^{\pi
^{\ast }},t)}{X_{t}^{\pi ^{\ast }}U_{xx}(X_{t}^{\pi ^{\ast }},t)}-\frac{%
a_{x}(X_{t}^{\pi ^{\ast }},t)}{X_{t}^{\pi ^{\ast }}U_{xx}(X_{t}^{\pi ^{\ast
}},t)}\right) ,  \label{optimal-feedback}
\end{equation}%
is used, then the control process $\pi _{t}^{\ast }$ is optimal. We note
that these arguments are formal and a general verification theorem is still
lacking.

Herein, we bypass these difficulties and construct homothetic forward
performance processes in factor-form using directly the Markovian solutions
of associated ergodic BSDE. The SPDE\ is merely used to guess the
appropriate form of the latter.

\section{Power case}

\label{SectionPowerUtility}

We start with the construction of forward performance processes that are
homogeneous of degree $\delta \in (0,1)$, and have the factor-form%
\begin{equation}
U\left( x,t\right) =\frac{x^{\delta }}{\delta }e^{f(V_{t},t)},
\label{power-general}
\end{equation}%
where $f:\mathbb{R}^{d}\times \lbrack 0,\infty )\rightarrow \mathbb{R}$ is a
(deterministic) function to be specified. For this range of $\delta ,$ the
admissible wealth domain is taken to be $D=\mathbb{R}_{+}.$

Using the form (\ref{power-general}) and the SPDE (\ref{SPDE}), we deduce
that $f$ must satisfy, for $\left( v,t\right) \in \mathbb{R}^{d}\times
\lbrack 0,\infty ),$ the semilinear\textbf{\ }PDE
\begin{equation}
f_{t}+\frac{1}{2}Trace\left( \kappa \kappa ^{T}\nabla ^{2}f\right) +\eta
(v)^{T}\nabla f+F(v,\kappa ^{T}\nabla f)=0,  \label{f-eqn}
\end{equation}%
with
\begin{equation}
F(v,z):=-\frac{1}{2}\delta (1-\delta )dist^{2}\left( \Pi ,\frac{z+\theta (v)%
}{1-\delta }\right) +\frac{1}{2}\frac{\delta }{1-\delta }|z+\theta (v)|^{2}+%
\frac{1}{2}|z|^{2}.  \label{driver-formal}
\end{equation}%
The above equation, however, is ill-posed with no known solutions to date.
On the other hand, as we demonstrate below, the process $f\left(
V_{t},t\right) $ can be actually constructed directly from the Markovian
solution of an ergodic BSDE whose driver is of the above form (cf. (\ref%
{driver-power})).

\subsection{Construction via ergodic BSDE}

We firstly introduce the underlying ergodic BSDE and provide the main
existence and uniqueness result for Markovian solutions. For the reader's
convenience, we present the proof in the Appendix.

\begin{proposition}
\label{Theorem1_EQBSDE}

Assume that the market price of risk vector $\theta \left( v\right) $
satisfies Assumption 1.ii and let the set $\Pi $ be as in (\ref%
{admissibility}). Then, the ergodic BSDE
\begin{equation}
dY_{t}=(-F(V_{t},Z_{t})+\lambda )dt+Z_{t}^{T}dW_{t},  \label{EQBSDE1}
\end{equation}%
with the driver $F(\cdot ,\cdot )$ defined as
\begin{equation}
F(V_{t},Z_{t}):=-\frac{1}{2}\delta (1-\delta )dist^{2}\left( \Pi ,\frac{%
Z_{t}+\theta (V_{t})}{1-\delta }\right) +\frac{1}{2}\frac{\delta }{1-\delta }%
|Z_{t}+\theta (V_{t})|^{2}+\frac{1}{2}|Z_{t}|^{2},  \label{driver-power}
\end{equation}%
admits a unique Markovian solution $(Y_{t},Z_{t},\lambda ),$ $t\geq 0.$

Specifically, there exist a unique $\lambda \in \mathbb{R}$ and functions $y:%
\mathbb{R}^{d}\rightarrow \mathbb{R}$ and $z:\mathbb{R}^{d}\rightarrow
\mathbb{R}^{d}$ such that $\left( Y_{t},Z_{t}\right) =\left( y\left(
V_{t}\right) ,z\left( V_{t}\right) \right) $. The function $y(\cdot )$ is
unique up to a constant and has at most linear growth, and $z(\cdot )$ is
bounded with $|z(\cdot )|\leq \frac{C_{v}}{C_{\eta }-C_{v}}$, where $C_{\eta
}$ and $C_{v}$ are as in (\ref{dissipative}) and (\ref{driver0}),
respectively.
\end{proposition}

We next present one of the main results.

\begin{theorem}
\label{Theorem2_ForwardUtility}

Let $(Y_{t},Z_{t},\lambda )=(y(V_{t}),z(V_{t}),\lambda ),t\geq 0,$ be the
unique Markovian solution of (\ref{EQBSDE1}). Then,

i) the process $U(x,t),$ $\left( x,t\right) \in \mathbb{R}_{+}\times \left[
0,\infty \right) ,$ given by
\begin{equation}
U(x,t)=\frac{x^{\delta }}{\delta }e^{y(V_{t})-\lambda t}\text{ ,}
\label{PowerForwardUtility}
\end{equation}%
is a power forward performance process with volatility%
\begin{equation}
a\left( x,t\right) =\frac{x^{\delta }}{\delta }e^{y(V_{t})-\lambda
t}z(V_{t}).  \label{power-volatility}
\end{equation}

ii) The optimal portfolio weights $\pi _{t}^{\ast }$ and the associated
wealth process $X_{t}^{\ast }$ (cf. (\ref{policy-normalized}) and (\ref%
{wealth-process})) are given, respectively, by
\begin{equation}
\pi _{t}^{\ast }=Proj_{\Pi }\left( \frac{z(V_{t})+\theta (V_{t})}{1-\delta }%
\right) \text{ \ and }X_{t}^{\ast }=X_{0}\mathcal{E}\left( \int_{0}^{\cdot
}(\pi _{s}^{\ast })^{T}(\theta (V_{s})ds+dW_{s})\right) _{t}.
\label{PowerOptimalStrategy}
\end{equation}
\end{theorem}

\begin{proof}
It is immediate that the process $U(x,t)$ is $\mathbb{F}$-progressively
measurable, strictly increasing and strictly concave in $x,$ and homogeneous
of degree $\delta .$ To show that it also satisfies requirements \textit{%
(ii) }and \textit{(iii)} of Definition 1, we will establish that, for $0\leq
t\leq s$, if $\pi \in \mathcal{A}$,
\begin{equation*}
E_{\mathbb{P}}\left( \frac{(X_{s}^{\pi })^{\delta }}{\delta }%
e^{Y_{s}-\lambda s}|\mathcal{F}_{t}\right) \leq \frac{(X_{t}^{\pi })^{\delta
}}{\delta }e^{Y_{t}-\lambda t},
\end{equation*}%
while for $\pi ^{\ast }$ given by (\ref{PowerOptimalStrategy}),
\begin{equation*}
E_{\mathbb{P}}\left( \frac{(X_{s}^{\pi ^{\ast }})^{\delta }}{\delta }%
e^{Y_{s}-\lambda s}|\mathcal{F}_{t}\right) =\frac{(X_{t}^{\pi ^{\ast
}})^{\delta }}{\delta }e^{Y_{t}-\lambda t}.
\end{equation*}%
To this end, the wealth equation (\ref{wealth-process})\ and It\^{o}'s
formula yield%
\begin{equation*}
(X_{s}^{\pi })^{\delta }=(X_{t}^{\pi })^{\delta }\exp \left(
\int_{t}^{s}\delta \left( \pi _{u}^{T}\theta (V_{u})-\frac{1}{2}|\pi
_{u}|^{2}\right) du+\int_{t}^{s}\delta \pi _{u}^{T}dW_{u}\right) .
\end{equation*}%
On the other hand, from the ergodic BSDE (\ref{EQBSDE1}), we have
\begin{equation}
Y_{s}-\lambda s=Y_{t}-\lambda
t-\int_{t}^{s}F(V_{u},z(V_{u}))du+\int_{t}^{s}z(V_{u})^{T}dW_{u}.
\label{Y-solution}
\end{equation}%
Combining the above yields
\begin{equation*}
(X_{s}^{\pi })^{\delta }e^{Y_{s}-\lambda s}=\ (X_{t}^{\pi })^{\delta
}e^{Y_{t}-\lambda t}\exp \left( \int_{t}^{s}\left( \delta \left( \pi
_{u}^{T}\theta (V_{u})-\frac{1}{2}|\pi _{u}|^{2}\right)
-F(V_{u},z(V_{u}))\right) du\right.
\end{equation*}%
\begin{equation*}
+\left. \int_{t}^{s}\left( \delta \pi _{u}^{T}+z(V_{u})^{T}\right)
dW_{u}\right) .
\end{equation*}%
Therefore,%
\begin{equation*}
E_{\mathbb{P}}\left( (X_{s}^{\pi })^{\delta }e^{Y_{s}-\lambda s}|\mathcal{F}%
_{t}\right)
\end{equation*}%
\begin{equation*}
=\ (X_{t}^{\pi })^{\delta }e^{Y_{t}-\lambda t}E_{\mathbb{P}}\left( \exp
\left( \int_{t}^{s}\left( \delta \left( \pi _{u}^{T}\theta (V_{u})-\frac{1}{2%
}|\pi _{u}|^{2}\right) -F(V_{u},z(V_{u}))\right) du\right. \right.
\end{equation*}%
\begin{equation*}
\left. \left. +\int_{t}^{s}\left( \delta \pi _{u}^{T}+z(V_{u})^{T}\right)
dW_{u}\right \vert \mathcal{F}_{t}\right) .
\end{equation*}%
Next, for $s\geq 0$ and $\pi \in \mathcal{A}$, we define a probability
measure, say $\mathbb{Q}^{\pi },$ by introducing the Radon-Nikodym density
process $\mathcal{Z}_{u},u\in \left[ 0,s\right] ,$
\begin{equation}
\mathcal{Z}_{u}=\left. \frac{d\mathbb{Q}^{\pi }}{d\mathbb{P}}\right \vert _{%
\mathcal{F}_{u}}=\mathcal{E}(N)_{u}\text{ \ with \ }N_{u}=\int_{0}^{u}\left(
\delta \pi _{t}^{T}+Z_{t}^{T}\right) dW_{t}.  \label{missing}
\end{equation}%
We recall that both processes $\pi _{u}$ and $z(V_{u})$, $u\in \left[ 0,s%
\right] ,$ satisfy the \textit{BMO}-condition (up to time $s$). Therefore,
the process $N_{u},$ $u\in \left[ 0,s\right] ,$ is a\textit{\ BMO}%
-martingale and, in turn, $\mathcal{E}(N)$ is in Doob's class $\mathcal{D}$
and, thus, uniformly integrable. In turn,%
\begin{equation*}
E_{\mathbb{P}}\left( \left. \exp \left( \int_{t}^{s}\left( F^{\pi
}(V_{u},z(V_{u}))-F(V_{u},z(V_{u}))\right) du\right) \frac{\mathcal{Z}_{s}}{%
\mathcal{Z}_{t}}\right \vert \mathcal{F}_{t}\right)
\end{equation*}%
\begin{equation*}
=\ E_{\mathbb{Q}^{\pi }}\left( \left. \exp \left( \int_{t}^{s}\left( F^{\pi
}(V_{u},z(V_{u}))-F(V_{u},z(V_{u}))\right) du\right) \right \vert \mathcal{F}%
_{t}\right) ,
\end{equation*}%
where
\begin{equation*}
F^{\pi }(V_{t},z(V_{t})):=-\frac{1}{2}\delta (1-\delta )|\pi
_{t}|^{2}+\delta \pi _{t}^{T}(z(V_{t})+\theta (V_{t}))+\frac{1}{2}%
|z(V_{t})|^{2}
\end{equation*}%
\begin{equation*}
=-\frac{1}{2}\delta (1-\delta )\left \vert \pi _{t}-\frac{z(V_{t})+\theta
(V_{t})}{1-\delta }\right \vert ^{2}+\frac{1}{2}\frac{\delta }{1-\delta }%
|z(V_{t})+\theta (V_{t})|^{2}+\frac{1}{2}|z(V_{t})|^{2}.
\end{equation*}%
Using that $F^{\pi }(V_{t},z(V_{t}))\leq F(V_{t},z(V_{t}))$, we easily
deduce that
\begin{equation*}
E_{\mathbb{P}}\left( (X_{s}^{\pi })^{\delta }e^{Y_{s}-\lambda s}|\mathcal{F}%
_{t}\right) \leq (X_{t}^{\pi })^{\delta }e^{Y_{t}-\lambda t}.
\end{equation*}%
Moreover, for $\pi =\pi ^{\ast }$ as in (\ref{PowerOptimalStrategy}), $%
F^{\pi ^{\ast }}(V_{t},z(V_{t}))=F(V_{t},z(V_{t}))$ and, thus,
\begin{equation*}
{E}_{\mathbb{P}}\left( (X_{s}^{\pi ^{\ast }})^{\delta }e^{Y_{s}-\lambda s}|%
\mathcal{F}_{t}\right) =(X_{t}^{\pi ^{\ast }})^{\delta }e^{Y_{t}-\lambda t}.
\end{equation*}%
To show (\ref{power-volatility}), we recall the SPDE\ (\ref{SPDE})\ and
observe that representation (\ref{PowerForwardUtility}) yields
\begin{equation*}
dU(x,t)=U(x,t)(-F(V_{t},z(V_{t}))+\frac{1}{2}%
|z(V_{t})|^{2})dt+U(x,t)z(V_{t})^{T}dW_{t}.
\end{equation*}%
The rest of the proof follows easily.
\end{proof}

\subsubsection{Connection with risk-sensitive optimization}

We provide an interpretation of the constant $\lambda ,$ appearing in the
representation of the forward performance process (\ref{PowerForwardUtility}%
), as the solution of the risk-sensitive control problem (\ref%
{ErgodicControlProblem}). It turns out that the constant $\lambda $ is also
the optimal growth rate of the long-term utility maximization problem as
considered in \cite{Bielecki}, \cite{FS2} and \cite{FS3} (see (\ref%
{ErgodicControlProblem1}) below).

\begin{proposition}
\label{propositionLambda} Let $T>0$ and ${\pi }\in \mathcal{A}$, and define
the probability measure ${\mathbb{P}}^{{\pi }}$ using the Radon-Nikodym
density process $\mathcal{Z}_{u}$, $u\in \lbrack 0,T]$,
\begin{equation}
\mathcal{Z}_{u}=\left. \frac{d{\mathbb{P}}^{{\pi }}}{d\mathbb{P}}\right
\vert _{\mathcal{F}_{u}}=\mathcal{E}\left( \int_{0}^{\cdot }\delta {\pi }%
_{u}^{T}dW_{u}\right) _{u}.  \label{RN_density_2}
\end{equation}%
and the stochastic functional
\begin{equation*}
L(V_{s},{\pi }_{s}):=-\frac{1}{2}\delta (1-\delta )|{\pi }_{s}|^{2}+\delta
\theta (V_{s})^{T}{\pi }_{s},
\end{equation*}
for $s\in \left[ 0,T\right] .$

Let $(y(V_{t}),z\left( V_{t}\right) ,\lambda ),$ $t\geq 0,$ be the unique
Markovian solution of the ergodic BSDE (\ref{EQBSDE1}) and $X^{\pi }$
solving the wealth equation (\ref{wealth-process}). Then, $\lambda $ is the
long-term growth rate of the risk-sensitive control problem
\begin{equation}
\lambda =\sup_{{\pi }\in \mathcal{A}}\limsup_{T\uparrow \infty }\frac{1}{T}%
\ln E_{{\mathbb{P}}^{{\pi }}}\left( e^{\int_{0}^{T}L(V_{s},{\pi }%
_{s})ds}\right) ,  \label{ErgodicControlProblem}
\end{equation}%
or, alternatively,
\begin{equation}
\lambda =\sup_{\pi \in \mathcal{A}}\limsup_{T\uparrow \infty }\frac{1}{T}\ln
E_{\mathbb{P}}\left( \frac{(X_{T}^{\pi })^{\delta }}{\delta }\right) .
\label{ErgodicControlProblem1}
\end{equation}%
For both problems (\ref{ErgodicControlProblem})\ and (\ref%
{ErgodicControlProblem1}), the associated optimal control process $\pi
_{t}^{\ast }$, $t\geq 0$, is as in (\ref{PowerOptimalStrategy}).
\end{proposition}

\begin{proof}
We first observe that the driver $F(\cdot ,\cdot )$ in (\ref{driver-power})
can be written as
\begin{equation*}
F(V_{t},Z_{t})=\sup_{{\pi }_{t}\in \Pi }\left( L(V_{t},\pi
_{t})+Z_{t}^{T}\delta {\pi }_{t}\right) +\frac{1}{2}|Z_{t}|^{2}.
\end{equation*}%
Therefore, for arbitrary $\tilde{\pi}\in \mathcal{A}$, we rewrite the
ergodic BSDE (\ref{EQBSDE1}) under the probability measure ${\mathbb{P}}^{%
\tilde{\pi}}$ as
\begin{equation*}
dY_{t}=\left( -\sup_{{\pi }_{t}\in \Pi }\left( L(V_{t},\pi
_{t})+Z_{t}^{T}\delta {\pi }_{t}\right) +Z_{t}^{T}\delta \tilde{\pi}%
_{t}+\lambda -\frac{1}{2}|Z_{t}|^{2}\right) dt+Z_{t}^{T}dW_{t}^{{\mathbb{P}}%
^{\tilde{\pi}}},
\end{equation*}%
where the process $W_{t}^{{\mathbb{P}}^{\tilde{\pi}}}:=W_{t}-\int_{0}^{t}%
\delta \tilde{\pi}_{u}du$, $t\geq 0$, is a Brownian motion under ${\mathbb{P}%
}^{\tilde{\pi}}$. In turn,%
\begin{equation*}
e^{\lambda T+Y_{0}}e^{-Y_{T}}\mathcal{E}\left( \int_{0}^{\cdot
}Z_{u}^{T}dW_{u}^{\mathbb{P}^{\tilde{\pi}}}\right) _{T}
\end{equation*}%
\begin{equation*}
=\  \exp \left( \int_{0}^{T}\left( \sup_{{\pi }_{t}\in \Pi }\left(
L(V_{t},\pi _{t})+Z_{t}^{T}\delta {\pi }_{t}\right) -\left( L(V_{t},\tilde{%
\pi}_{t})+Z_{t}^{T}\delta \tilde{\pi}_{t}\right) \right) dt\right)
e^{\int_{0}^{T}L(V_{t},\tilde{\pi}_{t})dt}.
\end{equation*}%
Next, we observe that for any $\tilde{\pi}\in \mathcal{A}$, the first
exponential term on the right hand side is bounded below by $1$. Taking
expectation under $\mathbb{P}^{\tilde{\pi}}$ then yields
\begin{equation*}
e^{\lambda T+Y_{0}}E_{{\mathbb{P}}^{\tilde{\pi}}}\left( e^{-Y_{T}}\mathcal{E}%
\left( \int_{0}^{\cdot }Z_{s}^{T}dW_{s}^{\mathbb{P}^{\tilde{\pi}}}\right)
_{T}\right) \geq E_{{\mathbb{P}}^{\tilde{\pi}}}\left( e^{\int_{0}^{T}L(V_{s},%
\tilde{\pi}_{s})ds}\right) .
\end{equation*}%
Using the measure $\mathbb{Q}^{\tilde{\pi}},$ defined in (\ref{missing}), we
deduce that
\begin{equation*}
\lambda +\frac{Y_{0}}{T}+\frac{1}{T}\ln E_{\mathbb{Q}^{{\tilde{\pi}}}}\left(
e^{-Y_{T}}\right) \geq \frac{1}{T}\ln E_{{\mathbb{P}}^{\tilde{\pi}}}\left(
e^{\int_{0}^{T}L(V_{s},\tilde{\pi}_{s})ds}\right) .
\end{equation*}%
Note, however, that there exists a constant, say $C,$ independent of $T,$
such that
\begin{equation*}
\frac{1}{C}\leq E_{\mathbb{Q}^{\pi }}\left( e^{-Y_{T}}\right) \leq C.
\end{equation*}%
This follows from the linear growth property of the function $y\left( \cdot
\right) $ and the ergodicity condition (\ref{exponentialcondition}) (see,
for example, \cite{FMc}).

Sending $T\uparrow \infty $ then yields that, for any $\tilde{\pi}\in
\mathcal{A}$,
\begin{equation*}
\lambda \geq \limsup_{T\uparrow \infty }\frac{1}{T}\ln E_{{\mathbb{P}}^{%
\tilde{\pi}}}\left( e^{\int_{0}^{T}L(V_{s},\tilde{\pi}_{s})ds}\right) ,
\end{equation*}%
with equality choosing $\tilde{\pi}_{s}=\pi _{s}^{\ast }$, with $\pi
_{s}^{\ast }$ as in (\ref{PowerOptimalStrategy}).

To show that $\lambda $ also solves (\ref{ErgodicControlProblem1}), we
observe that for $\pi \in \mathcal{A}$, we have
\begin{equation*}
E_{\mathbb{P}}\left( \frac{(X_{T}^{\pi })^{\delta }}{\delta }\right) =\frac{%
X_{0}^{\delta }}{\delta }E_{\mathbb{P}}\left( e^{\int_{0}^{T}L(Vs,\pi
_{s})ds}\mathcal{E}\left( \int_{0}^{\cdot }\delta {\pi }_{s}^{T}dW_{s}%
\right) _{T}\right)
\end{equation*}%
\begin{equation*}
=\frac{X_{0}^{\delta }}{\delta }E_{\mathbb{P}^{\pi }}\left(
e^{\int_{0}^{T}L(V_{s},\pi _{s})ds}\right) ,
\end{equation*}%
and the rest of the arguments follow.
\end{proof}

\subsubsection{Connection with an ill-posed multi-dimensional semilinear PDE}

A by-product of the previous result is the construction of a smooth solution
to the ill-posed semilinear PDE given in (\ref{TimeReversedPDE}) below.
Recall that the latter was derived from (\ref{SPDE}) as a necessary
requirement when we seek forward processes of the form (\ref{power-general}%
). We establish below that for an appropriate initial datum, this ill-posed
PDE has a solution, which is separable in time and space.

We note that the well-posed analogue of this semilinear equation, as well as
of the one appearing in the exponential case (cf. (\ref{exponential-equation}%
)), have been extensively analyzed and used for the representation of
indifference prices, risk measures, power and exponential value functions,
and others. To our knowledge, however, their ill-posed versions have not
been studied, with the exception of the one-dimensional case studied in \cite%
{NZ}. This case, on the other hand, can be linearized and the solution is
constructed using an extension of Widder's theorem. We refer in detail to
this case in 3.1.3. However, the multidimensional case cannot be linearized
and, to our knowledge, no results for this case exist to date.

\begin{proposition}
Consider the ill-posed semilinear PDE
\begin{equation}
f_{t}+\mathcal{L}f+F(v,\kappa ^{T}\nabla f)=0,  \label{TimeReversedPDE}
\end{equation}%
$(v,t)\in \mathbb{R}^{d}\times \lbrack 0,\infty )$, with $F\left( \cdot
,\cdot \right) $ as in (\ref{driver-formal}) (or (\ref{driver-power})) and $%
\mathcal{L}$ being the infinitesimal generator of the factor process $V$,
\begin{equation}
\mathcal{L}=\frac{1}{2}\text{Trace}\left( \kappa \kappa ^{T}\nabla
^{2}\right) +\eta (v)^{T}\nabla .  \label{generatorofV}
\end{equation}%
For initial condition $f(v,0)=y(v)$, where $y(\cdot )$ is the function
appearing in the Markovian solution $\left( y\left( V_{t}\right) ,z\left(
V_{t}\right) ,\lambda \right) $ of the ergodic BSDE (\ref{EQBSDE1}),
equation (\ref{TimeReversedPDE}) admits \textit{a smooth solution given by}
\begin{equation*}
f(v,t)=y(v)-\lambda t.
\end{equation*}
\end{proposition}

\begin{proof}
Firstly, assume that the function $y\left( \cdot \right) $ appearing in
Proposition 2 is in $C^{2}(\mathbb{R}^{d})$. It\^{o}'s formula then gives
\begin{equation*}
dy(V_{t})=\mathcal{L}y(V_{t})dt+\left( \kappa ^{T}\nabla y(V_{t})\right)
^{T}dW_{t},
\end{equation*}%
which combined with (\ref{EQBSDE1}) yields that $Z_{t}=z\left( V_{t}\right)
=\kappa ^{T}\nabla y(V_{t})$ and
\begin{equation*}
-\lambda +\mathcal{L}y(V_{t})+F(V_{t},\kappa ^{T}\nabla y(V_{t}))=0.
\end{equation*}%
It, therefore, remains to show that $y\left( \cdot \right) \in C^{2}(\mathbb{%
R}^{d})$. Indeed, for any $\rho >0$, consider the semilinear elliptic PDE
\begin{equation}
\rho y^{\rho }=\mathcal{L}y^{\rho }+F\left( v,\kappa ^{T}\nabla {y}^{\rho
}\right) .  \label{elliptic}
\end{equation}%
Classical PDE results yield that the above equation admits a unique bounded
solution $y^{\rho }\left( \cdot \right) \in C^{2}(\mathbb{R}^{d})$. Using
arguments similar to the ones in the Appendix, we deduce that $|y^{\rho
}(v)|\leq \frac{K}{\rho }$ and $|\nabla y^{\rho }(v)|\leq \frac{C_{v}}{%
C_{\eta }-C_{v}}.$

Therefore, for any reference point, say $v^{0}\in \mathbb{R}^{d}$, we have
that $\rho y^{\rho }(v_{0})$ is uniformly bounded and, moreover, that the
difference $y^{\rho }(v)-y^{\rho }(v_{0})$ is equicontinuous. Using a
diagonal argument (cf. (\ref{**}) in the Appendix), we deduce that there
exists a subsequence $\rho _{n}\downarrow 0$ such that $\rho _{n}y^{\rho
_{n}}(v_{0})\rightarrow \lambda $ and $y^{\rho _{n}}(v)-y^{\rho
_{n}}(v_{0})\rightarrow y(v),$ uniformly on compact sets of $\mathbb{R}^{d}$%
. Since, however, both $\rho _{n}y^{\rho _{n}}(v)$ and $\nabla y^{\rho
_{n}}(v)$ are bounded uniformly in $\rho _{n}$, $\nabla ^{2}y^{\rho _{n}}(v)$
is also bounded on compact sets, as it follows from equation (\ref{elliptic}%
) above. In turn, this yields a H\"{o}lder estimate for $\nabla y^{\rho
_{n}}(v),$ uniformly on compact sets. Standard arguments for elliptic
equations then give that the limit $y(\cdot )\in C^{2}(\mathbb{R}^{d})$
(see, for example, Theorem 3.3 of \cite{FMc}).
\end{proof}

\subsubsection{Example: Single stock and single stochastic factor}

\label{SectionExample}

For the state equations (\ref{stock-SDE}) and (\ref{factor-SDE}), let $n=1$
and $d=2$. Then, the stock and the stochastic factor processes follow,
respectively,
\begin{equation*}
dS_{t}=b(V_{t})S_{t}dt+\sigma (V_{t})S_{t}dW_{t}^{1}\text{,}
\end{equation*}%
\begin{equation*}
dV_{t}^{1}=\eta (V_{t})dt+\kappa ^{1}dW_{t}^{1}+\kappa ^{2}dW_{t}^{2}\text{
\  \ and \ }dV_{t}^{2}=0,
\end{equation*}%
with $\min (\kappa ^{1},\kappa ^{2})>0$, $|\kappa ^{1}|^{2}+|\kappa
^{2}|^{2}=1$ and $\sigma \left( \cdot \right) $ bounded by a positive
constant.

Let $\Pi =\mathbb{R\times }\left \{ 0\right \} $ so that $\pi _{t}^{2}\equiv
0. $ Then, the wealth equation (\ref{wealth-process}) reduces to $%
dX_{t}^{\pi }=X_{t}^{\pi }\pi _{t}^{1}\left( \theta
(V_{t})dt+dW_{t}^{1}\right) $ with $\theta (V_{t})=b(V_{t})/\sigma (V_{t}).$
In turn, the driver of (\ref{EQBSDE1}) takes the form
\begin{equation*}
F(V_{t},Z_{t}^{1},Z_{t}^{2})=\  \frac{1}{2}\frac{\delta }{1-\delta }%
|Z_{t}^{1}+\theta (V_{t})|^{2}+\frac{1}{2}|Z_{t}^{1}|^{2}+\frac{1}{2}%
|Z_{t}^{2}|^{2}.
\end{equation*}%
By Theorem \ref{Theorem2_ForwardUtility}, the optimal portfolio weights are $%
\left( \pi _{t}^{\ast ,1},\pi _{t}^{\ast ,2}\right) =\left( \frac{%
Z_{t}^{1}+\theta ^{1}(V_{t})}{1-\delta },0\right) $.

Next, note that if $C_{\theta }$ and $K_{\theta }$ are, respectively, the
Lipschitiz constant and the bound for the market price of risk $\theta (v)$
(cf. Assumption 1.\textit{ii}), then
\begin{equation*}
|F(v,z^{1},z^{2})-F(\bar{v},z^{1},z^{2})|\leq \frac{\delta }{1-\delta }%
|z^{1}+\theta (v)||\theta (v)-\theta (\bar{v})|
\end{equation*}%
\begin{equation*}
\leq \frac{\delta }{1-\delta }\max \{1,K_{\theta }\}C_{\theta }(1+|z|)|v-%
\bar{v}|.
\end{equation*}%
Hence, we may take in inequality (\ref{driver0}) the constant $C_{v}$ to be
defined as $C_{v}=\delta \max \{1,K_{\theta }\}C_{\theta }/(1-\delta )$.

To find the processes $Z_{t}^{1}$ and $Z_{t}^{2},$ we set $Z_{t}^{i}=\kappa
^{i}Z_{t},$ $i=1,2$, for some process $Z_{t}$ to be determined. Then,
equation (\ref{EQBSDE1}) further reduces to%
\begin{equation*}
dY_{t}=\left( -\frac{\hat{\delta}}{2}|Z_{t}|^{2}-\frac{\delta \kappa ^{1}}{%
1-\delta }\theta (V_{t})Z_{t}-\frac{\delta }{2(1-\delta )}|\theta
(V_{t})|^{2}+\lambda \right) dt
\end{equation*}%
\begin{equation*}
+Z_{t}\left( \kappa ^{1}dW_{t}^{1}+\kappa ^{2}dW_{t}^{2}\right) ,
\end{equation*}%
with $\hat{\delta}=\frac{1-\delta +\delta |\kappa ^{1}|^{2}}{1-\delta }.$

Next, let $\tilde{Y}_{t}:=e^{\hat{\delta}(Y_{t}-\lambda t)}\  \ $and $\
\tilde{Z}_{t}:=\hat{\delta}\tilde{Y}_{t}Z_{t}.$ Then,
\begin{equation*}
d\tilde{Y}_{t}=-\hat{\delta}\frac{\delta }{2(1-\delta )}|\theta (V_{t})|^{2}%
\tilde{Y}_{t}dt+\tilde{Z}_{t}d\tilde{W}_{t},
\end{equation*}%
where $\tilde{W}_{t}:=\kappa ^{1}W_{t}^{1}+\kappa ^{2}W_{t}^{2}-\int_{0}^{t}%
\frac{\delta \kappa ^{1}}{1-\delta }\theta (V_{u})du$, $t\geq 0$, is a
Brownian motion under some probability measure equivalent to $\mathbb{P}$.

Let $\beta _{t}:=\exp \left( \int_{0}^{t}\hat{\delta}\frac{\delta }{%
2(1-\delta )}|\theta (V_{u})|^{2}du\right) .$ Applying It\^{o}'s formula to $%
\tilde{Y}_{t}\beta _{t}$ yields
\begin{equation*}
\tilde{Y}_{t}=\frac{\beta _{0}}{\beta _{t}}\tilde{Y_{0}}+\int_{0}^{t}\frac{%
\beta _{u}}{\beta _{t}}\tilde{Z}_{u}d\tilde{W}_{u}.
\end{equation*}%
The power forward performance process can be then written as
\begin{equation*}
U(x,t)=\frac{x^{\delta }}{\delta }(\tilde{Y}_{t})^{{1}/{\hat{\delta}}}=\frac{%
x^{\delta }}{\delta }\left( \frac{\beta _{0}}{\beta _{t}}\tilde{Y_{0}}%
+\int_{0}^{t}\frac{\beta _{u}}{\beta _{t}}\tilde{Z}_{u}d\tilde{W}_{u}\right)
^{{1}/{\hat{\delta}}}.
\end{equation*}%
The above result yields an alternative representation to the solution
derived in \cite{NZ}, where the same market model is considered, bypassing
various lengthy steps for the reduced linearized forward SPDE. Indeed, one
can easily deduce that writing $\tilde{Y}_{t}=\tilde{y}\left( V_{t},t\right)
$ and using the dynamics of the stochastic factor (\ref{factor-SDE}), yields
that $\tilde{y}\left( v,t\right) $ must satisfy%
\begin{equation*}
\tilde{y}_{t}\left( v,t\right) +\frac{1}{2}\tilde{y}_{vv}\left( v,t\right)
+\left( \eta \left( v\right) +\frac{\delta \kappa ^{1}}{1-\delta }\theta
\left( v\right) \right) \tilde{y}_{v}\left( v,t\right) +\frac{\hat{\delta}%
\delta }{2\left( 1-\delta \right) }\theta ^{2}\left( v\right) \tilde{y}%
\left( v,t\right) =0,
\end{equation*}%
recovering directly the result of \cite{NZ}.

\subsection{Connection with infinite horizon BSDE}

\label{SubSectionPowerUtility}

In this section, we build a connection between power forward processes in
factor-form and non-Markovian solutions of a family of infinite horizon
BSDE. The contribution is threefold.

Firstly, these solutions are themselves power forward processes, albeit not
in a factor-form. Secondly, we consider their limit as the parameter $\rho ,$
appearing naturally in the infinite horizon BSDE, converges to zero. We
establish that appropriately discounted, they provide an approximation to
the process $U\left( x,t\right) $ as $\rho \downarrow 0.$ Thirdly, we build
a connection with a family of classical value function processes in finite
horizon, say $[0,T],$ when the horizon is long $\left( T\uparrow \infty
\right) $.

We start with some background results on infinite horizon BSDE. Among
others, we recall that \cite{Briand} is one of the first papers in which
Girsanov's transformation is used to solve infinite horizon BSDE with
Lipschitz driver, while the quadratic driver case was solved in \cite%
{Briand0}. We refer the reader to \cite{Briand0} for further references.

\begin{proposition}
Let $\rho >0$, and consider the infinite horizon BSDE
\begin{equation}
dY_{t}^{\rho }=\left( -F(V_{t},Z_{t}^{\rho })+\rho Y_{t}^{\rho }\right)
dt+\left( Z_{t}^{\rho }\right) ^{T}dW_{t},  \label{IQBSDE1}
\end{equation}%
where the driver $F(\cdot ,\cdot )$ is given in (\ref{EQBSDE1}), with $%
\theta \left( \cdot \right) ,$ $\Pi $ and $V$ satisfying the assumptions in
section 1. Then, equation (\ref{IQBSDE1}) admits a unique Markovian solution
$\left( Y_{t}^{\rho },Z_{t}^{\rho }\right) ,$ $t\geq 0.$

Specifically, for each $\rho >0,$ there exist unique functions $y^{\rho }:%
\mathbb{R}^{d}\rightarrow \mathbb{R}$ and $z^{\rho }:\mathbb{R}%
^{d}\rightarrow \mathbb{R}^{d}$ such that $(Y_{t}^{\rho },Z_{t}^{\rho
})=(y^{\rho }(V_{t}),z^{\rho }(V_{t})),$ with $|y^{\rho }\left( \cdot
\right) |\leq \frac{K}{\rho }$ and $|z^{\rho }\left( \cdot \right) |\leq
\frac{C_{v}}{C_{\eta }-C_{v}},$ where $C_{\eta }$ as in (\ref{dissipative}),
and $C_{v}$, $K$ given in (\ref{driver0}) and (\ref{driver2}), respectively.
\end{proposition}

The solvability of (\ref{IQBSDE1}) is an intermediate step to solve (\ref%
{EQBSDE1}), and is included in the proof of Proposition \ref{Theorem1_EQBSDE}
in the Appendix.

\begin{theorem}
\label{Theorem21_ForwardUtility} Let $\left( y^{\rho }\left( V_{t}\right)
,z^{\rho }\left( V_{t}\right) \right) ,$ $t\geq 0,$ be the unique Markovian
solution to the infinite horizon BSDE (\ref{IQBSDE1}). Then,

i) the process $U^{\rho }\left( x,t\right) ,$ $\left( x,t\right) \in \mathbb{%
R}_{+}\times \left[ 0,\infty \right) ,$ given by%
\begin{equation}
U^{\rho }(x,t)=\frac{x^{\delta }}{\delta }e^{y^{\rho }\left( V_{t}\right)
-\int_{0}^{t}\rho y^{\rho }\left( V_{s}\right) ds}
\label{PowerForwardUtility2}
\end{equation}%
is a power forward performance process with volatility
\begin{equation*}
a^{\rho }\left( x,t\right) =\frac{x^{\delta }}{\delta }e^{y^{\rho }\left(
V_{t}\right) -\int_{0}^{t}\rho y^{\rho }\left( V_{s}\right) ds}z^{\rho
}\left( V_{t}\right) .
\end{equation*}

ii) The optimal portfolio weights $\pi _{t}^{\ast ,\rho }$ and the
associated wealth process $X_{t}^{\ast ,\rho }$ (cf. (\ref{policy-normalized}%
),(\ref{wealth-process})), $t\geq 0,$ are given, respectively, by
\begin{equation*}
\pi _{t}^{\ast ,\rho }=Proj_{\Pi }\left( \frac{z^{\rho }(V_{t})+\theta
(V_{t})}{1-\delta }\right) \text{ \ and \ }X_{t}^{\ast ,\rho }=X_{0}\mathcal{%
E}\left( \int_{0}^{\cdot }(\pi _{s}^{\ast ,\rho })^{T}(\theta
(V_{s})ds+dW_{s})\right) _{t}.
\end{equation*}
\end{theorem}

The proof is similar to the one of Theorem \ref{Theorem2_ForwardUtility},
and it is thus omitted.

The next result relates the factor-from forward process $U\left( x,t\right) $
(cf. Theorem \ref{Theorem2_ForwardUtility}) and the path-dependent one $%
U^{\rho }\left( x,t\right) $ (cf. Theorem \ref{Theorem21_ForwardUtility}),
and their corresponding optimal portfolio strategies.

We use the superscript $v$ to denote dependence on the initial condition.

\begin{proposition}
\label{PropositionPowerUtility} For $\left( x,t\right) \in \mathbb{R}%
_{+}\times \left[ 0,\infty \right) ,$ let $U(x,t)$ and $U^{\rho }(x,t)$ be
the power forward processes given in (\ref{PowerForwardUtility}) and (\ref%
{PowerForwardUtility2}), and $y\left( V_{t}\right) $ the component of the
Markovian solution to the ergodic BSDE (\ref{EQBSDE1}). Then, for an
arbitrary reference point $v_{0}\in \mathbb{R}^{d}$, there exists a
subsequence $\rho _{n}\downarrow 0$ (depending on $v_{0})$ such that, for $%
\left( x,t\right) \in \mathbb{R}_{+}\times \left[ 0,\infty \right) ,$
\begin{equation}
\lim_{\rho _{n}\downarrow 0}\frac{U^{\rho _{n}}(x,t)e^{-y^{\rho _{n}}\left(
v_{0}\right) }}{U(x,t)}=1.  \label{relationU}
\end{equation}%
Moreover, for each $t\geq 0,$ the associated optimal portfolio weights $\pi
^{\ast ,\rho _{n}}$ and $\pi ^{\ast }$ satisfy
\begin{equation}
\lim_{\rho _{n}\downarrow 0}E_{\mathbb{P}}\int_{0}^{t}\left \vert \pi
_{s}^{\ast ,\rho _{n}}-\pi _{s}^{\ast }\right \vert ^{2}ds=0.
\label{portfolios}
\end{equation}
\end{proposition}

\begin{proof}
For an arbitrary reference point $v_{0}\in \mathbb{R}^{d}$, from the
representations (\ref{PowerForwardUtility}) and (\ref{PowerForwardUtility2}%
), we have that
\begin{equation*}
\frac{U^{\rho }(x,t)e^{-y^{\rho }\left( v_{0}\right) }}{U(x,t)}=\exp \left(
(y^{\rho }\left( V_{t}^{v}\right) -\int_{0}^{t}\rho y^{\rho }\left(
V_{u}^{v}\right) du)-(y(V_{t}^{v})-\lambda t)-y^{\rho }(v_{0}\right)
\end{equation*}%
\begin{equation*}
=\exp \left( \left( y^{\rho }\left( V_{t}^{v}\right) -y^{\rho }\left(
v_{0}\right) -y(V_{t}^{v})\right) -\int_{0}^{t}\rho \left( y^{\rho }\left(
V_{u}^{v}\right) -y^{\rho }\left( v_{0}\right) \right) du-(\rho y^{\rho
}\left( v_{0}\right) -\lambda )t\right) .
\end{equation*}%
On the other hand, the limits (\ref{**}) and (\ref{*}), established in the
Appendix, yield that there exists a subsequence $\rho _{n}\downarrow 0$ such
that%
\begin{equation*}
\lim_{\rho _{n}\downarrow 0}\left( y^{\rho _{n}}\left( V_{t}^{v}\right)
-y^{\rho _{n}}\left( v_{0}\right) -y(V_{t}^{v}\right) )=0,\text{ \  \ }
\end{equation*}%
\begin{equation*}
\lim_{\rho _{n}\downarrow 0}\rho _{n}\left( y^{\rho _{n}}\left(
V_{t}^{v}\right) -y^{\rho _{n}}\left( v_{0}\right) \right) =0\text{ \  \  \  \
and \  \ }\lim_{\rho _{n}\downarrow 0}\left( \rho _{n}y^{\rho _{n}}\left(
v_{0}\right) -\lambda \right) =0,
\end{equation*}%
and we conclude.

To show assertion (\ref{portfolios}), we use the Lipschitz continuity of the
projection operator on the convex set $\Pi $ and the convergence%
\begin{equation}
\lim_{\rho _{n}\downarrow 0}E_{\mathbb{P}}\int_{0}^{t}\left \vert z^{\rho
_{n}}\left( V_{s}^{v}\right) -z\left( V_{s}^{v}\right) \right \vert ^{2}ds=0,
\label{Z-claim}
\end{equation}%
for $t\geq 0$. The latter is established in the Appendix.
\end{proof}

\subsection{Connection with the classical power expected utility for long
horizons}

\label{SubSectionTraditionalPowerUtility}

We examine whether the forward processes $U\left( x,t\right) $ and $U^{\rho
}\left( x,t\right) $ can be interpreted as long-term limits of the classical
value function process. We show that this is indeed the case for a family of
expected utility models with appropriately chosen terminal random
(multiplicative) payoffs.

To this end, let $\left[ 0,T\right] $ be an arbitrary trading horizon and
introduce, for $\rho >0,$ the value function process
\begin{equation}
u^{\rho }(x,t;T)=ess\sup_{\pi \in \mathcal{A}_{[t,T]}}E_{\mathbb{P}}\left(
\frac{(X_{T}^{\pi }e^{\xi _{T}})^{\delta }}{\delta }|\mathcal{F}%
_{t},X_{t}^{\pi }=x\right) ,  \label{PowerUtility}
\end{equation}%
for $\left( x,t\right) \in \mathbb{R}_{+}\times \left[ 0,T\right] $ and the
wealth process $X_{s}^{\pi }$, $s\in \lbrack t,T]$ solving (\ref%
{wealth-process}).

The payoff $\xi _{T}$ is defined as
\begin{equation}
\xi _{T}:=-\frac{1}{\delta }\int_{0}^{T}\rho Y_{t}^{\rho ,T}dt,
\label{endowment}
\end{equation}%
where $Y_{t}^{\rho ,T}$ is the solution of the finite-horizon quadratic BSDE
\begin{equation}
Y_{t}^{\rho ,T}=\int_{t}^{T}\left( F(V_{s},Z_{s}^{\rho ,T})-\rho Y_{s}^{\rho
,T}\right) ds-\int_{t}^{T}\left( Z_{s}^{\rho ,T}\right) ^{T}dW_{s},
\label{QBSDEPower}
\end{equation}%
with the driver $F(\cdot ,\cdot )$ given in (\ref{driver-power}). The
associated optimal portfolio weights are denoted by $\pi _{s}^{\ast ,\rho
,T} $, $s\in \lbrack t,T]$.

We recall that the classical optimal investment problem with power utility
has been solved using quadratic BSDE methods in \cite{HU0} for a Brownian
motion setting, and in \cite{Morlais} for a general semimartingale framework.

\begin{proposition}
\label{PropositionPowerUtility2} i) Let $u^{\rho }(x,t;T)$ and $U^{\rho
}(x,t)$ be given in (\ref{PowerUtility}) and (\ref{PowerForwardUtility2}),
respectively. Then, for each $\rho >0$ and $\left( x,t\right) \in \mathbb{R}%
_{+}\times \left[ 0,\infty \right) ,$
\begin{equation*}
\lim_{T\uparrow \infty }\frac{u^{\rho }(x,t;T)}{U^{\rho }(x,t)}=1,
\end{equation*}%
and the optimal portfolio weights satisfy, for $s\in \left[ t,T\right] ,$
\begin{equation*}
\lim_{T\uparrow \infty }E_{\mathbb{P}}\int_{t}^{s}\left \vert \pi _{u}^{\ast
,\rho ,T}-\pi _{u}^{\ast ,\rho }\right \vert ^{2}du=0.
\end{equation*}

ii) Let $U\left( x,t\right) $ be the power forward process as in (\ref%
{PowerForwardUtility}). Then, for each arbitrary reference point $v_{0}\in
\mathbb{R}^{d}$, there exists a subsequence $\rho _{n}\downarrow 0$
(depending on $v_{0}$) such that, for $\left( x,t\right) \in \mathbb{R}%
_{+}\times \left[ 0,\infty \right) ,$
\begin{equation*}
\lim_{\rho _{n}\downarrow 0}\lim_{T\uparrow \infty }\frac{u^{\rho
_{n}}(x,t;T)e^{-y^{\rho _{n}}\left( v_{0}\right) }}{U(x,t)}=1,
\end{equation*}%
and the optimal portfolio weights satisfy, for $s\in \left[ t,T\right] ,$
\begin{equation*}
\lim_{\rho _{n}\downarrow 0}\lim_{T\uparrow \infty }E_{\mathbb{P}%
}\int_{t}^{s}\left \vert \pi _{u}^{\ast ,\rho _{n},T}-\pi _{u}^{\ast }\right
\vert ^{2}du=0.
\end{equation*}
\end{proposition}

\begin{proof}
We only show part i). From Theorem 3.3 in \cite{Briand0} we have that $%
|Y_{t}^{\rho ,T}|\leq \frac{K}{\rho }$, and therefore, the quantity $\delta
\xi _{T}=-\int_{0}^{T}\rho Y_{u}^{\rho ,T}du$ is bounded. On the other hand,
the driver $F(\cdot ,\cdot )$ satisfies properties (\ref{driver1}) and (\ref%
{driver2}). Therefore, using similar arguments to the ones used in section 3
of \cite{HU0}, it follows that the value function process is given by $%
u^{\rho }(x,t;T)=\frac{x^{\delta }}{\delta }e^{Y_{t}}$, with $Y_{t}$ being
the \textit{unique} solution of the quadratic BSDE
\begin{equation}
Y_{t}=\delta \xi _{T}+\int_{t}^{T}F(V_{s},Z_{s})ds-\int_{t}^{T}\left(
Z_{s}\right) ^{T}dW_{s},  \label{QBSDE-power}
\end{equation}%
for $t\in \left[ 0,T\right] .$ In addition, the optimal portfolio weights
are given by $\pi _{t}^{\ast ,\rho ,T}=Proj_{\Pi }(\frac{Z_{t}+\theta (V_{t})%
}{1-\delta })$.

Note, however, that the pair of processes $(Y_{t}^{\rho ,T}-\int_{0}^{t}\rho
Y_{s}^{\rho ,T}ds,Z_{t}^{\rho ,T})$, $t\in \lbrack 0,T]$, with $(Y_{t}^{\rho
,T},Z_{t}^{\rho ,T})$ solving (\ref{QBSDEPower}), also satisfies the above
quadratic BSDE (\ref{QBSDE-power}). Therefore, we must have $%
Y_{t}=Y_{t}^{\rho ,T}-\int_{0}^{t}\rho Y_{s}^{\rho ,T}ds$, $t\in \left[ 0,T%
\right] $ and, as a consequence,
\begin{equation*}
u^{\rho }(x,t;T)=\frac{x^{\delta }}{\delta }\exp \left( Y_{t}^{\rho
,T}-\int_{0}^{t}\rho Y_{s}^{\rho ,T}ds\right) .
\end{equation*}%
In turn,
\begin{equation*}
\frac{u^{\rho }(x,t;T)}{U^{\rho }(x,t)}=\exp \left( (Y_{t}^{\rho
,T}-\int_{0}^{t}\rho Y_{s}^{\rho ,T}ds)-(Y_{t}^{\rho }-\int_{0}^{t}\rho
Y_{s}^{\rho }ds)\right) .
\end{equation*}%
Using (\ref{convergence0}) we deduce that $\lim_{T\uparrow \infty
}Y_{t}^{\rho ,T}=Y_{t}^{\rho }$, and we easily conclude.

The convergence of the optimal portfolio weights follows from the Lipschitz
continuity of the projection operator on the convex set $\Pi $ and the
convergence of $Z^{\rho ,T}$\ to $Z^{\rho }$\ in $\mathcal{L}_{\rho
}^{2}[t,\infty ).$ The space $\mathcal{L}_{\rho }^{2}[t,\infty )$ is defined
in (\ref{L-2-rho}) and the latter limit is shown in (\ref{convergence1}) in
the Appendix.
\end{proof}

\subsection{General (non-Markovian) power forward performance processes and
ergodic BSDE}

Departing from factor-form power forward performance processes, we may still
use the ergodic BSDE approach we developed earlier to construct such
processes of the general form%
\begin{equation*}
U\left( x,t\right) =\frac{x^{\delta }}{\delta }e^{K_{t}},
\end{equation*}%
for some $\mathbb{F}$-progressively measurable process $K_{t},t\geq 0,$
independent of $x.$

Indeed, consider an arbitrary process $Z\in \mathcal{L}_{BMO}^{2},$ and, in
turn, choose $(Y_{t},\lambda ),$ $t\geq 0,$ $\lambda \in \mathbb{R}$ and $Y$
being $\mathbb{F}$-progressively measurable, such that the triplet $%
(Y_{t},Z_{t},\lambda )$ solves the ergodic BSDE (\ref{EQBSDE1}). Using
similar arguments as the ones in the proof of Theorem 1, we may deduce that
the process
\begin{equation}
U\left( x,t\right) =\frac{x^{\delta }}{\delta }e^{Y_{t}-\lambda t},
\label{power}
\end{equation}%
$\left( x,t\right) \in \mathbb{R}_{+}\times \left[ 0,\infty \right) $
satisfies Definition 1. Then, the SPDE\ (\ref{SPDE})\ will yield that the
forward volatility is given by the process $a\left( x,t\right) =U\left(
x,t\right) Z_{t},t\geq 0.$ {One can also develop similar connections with
infinite horizon BSDE\ and the value function processes with terminal
(multiplicative) payoff,} as in sections 3.2 and 3.3.

The analysis of general power forward processes is beyond the scope of this
paper, and will be carried out separately. Herein, we only comment on three
examples, cast in the absence of portfolio constraints, $\Pi =\mathbb{R}%
^{d}. $

\textit{i) Time-monotone case}: Let $Z_{t}\equiv 0$, $t\geq 0$, and choose $%
(Y_{t},\lambda )$ as%
\begin{equation*}
Y_{t}-\lambda t:=Y_{0}-\int_{0}^{t}\frac{1}{2}\frac{\delta }{1-\delta }%
|\theta (V_{s})|^{2}ds,
\end{equation*}%
for any constant $Y_{0}\in \mathbb{R}$. Then, $(Y_{t},0,\lambda )$ satisfies
(\ref{EQBSDE1}). In turn, we deduce, using (\ref{power}) and the above, that
the process
\begin{equation*}
U(x,t):=e^{Y_{0}}\frac{x^{\delta }}{\delta }e^{-\frac{1}{2}\frac{\delta }{%
1-\delta }A_{t}},
\end{equation*}%
with \ $A_{t}=\int_{0}^{t}|\theta (V_{s})|^{2}ds,$ is a power performance.
This process has zero volatility $\left( a\left( x,t\right) \equiv 0\right)
, $ it is decreasing in time and path-dependent (see \cite{MZ2} for a
general study).

Variations of this solution with non-zero forward volatility can be
constructed, as it is shown below. We stress, however, that these forward
processes essentially correspond to a fictitious market with \textit{%
different risk-premia} and, thus, they do \textit{not} constitute genuine
new solutions for the original market.

\textit{ii) Market view case}: Let $Z_{t}=\phi _{t}$ with $\phi \in \mathcal{%
L}_{BMO}^{2}$, and choose $(Y_{t},\lambda ),$ $t\geq 0,$ as
\begin{equation*}
Y_{t}-\lambda t:=Y_{0}-\frac{1}{2}\frac{\delta }{1-\delta }\int_{0}^{t}|\phi
_{s}+\theta (V_{s})|^{2}ds-\frac{1}{2}\int_{0}^{t}|\phi
_{s}|^{2}ds+\int_{0}^{t}\phi _{s}^{T}dW_{s}.
\end{equation*}%
We can then verify that $\left( Y_{t},\phi _{t},\lambda \right) $ satisfies
equation (\ref{EQBSDE1}). Using (\ref{power}) and rearranging terms, we
deduce the representation
\begin{equation*}
U(x,t)=\frac{x^{\delta }}{\delta }e^{Y_{0}-\frac{1}{2}\frac{\delta }{%
1-\delta }\int_{0}^{t}|\phi _{s}+\theta (V_{s})|^{2}ds}\mathcal{E}\left(
\int_{0}^{\cdot }\phi _{s}^{T}dW_{s}\right) _{t}
\end{equation*}%
\begin{equation*}
=e^{Y_{0}}\frac{x^{\delta }}{\delta }e^{-\frac{1}{2}\frac{\delta }{1-\delta }%
A_{t}^{\phi }}M_{t}
\end{equation*}%
with $A_{t}^{\phi }=\int_{0}^{t}|\phi _{s}+\theta (V_{s})|^{2}ds$ and $M_{t}=%
\mathcal{E}(\int_{0}^{\cdot }\phi _{s}^{T}dW_{s})_{t}$.

\textit{iii) Benchmark case}: A different parametrization yields an
alternative representation and interpretation of the solution. Let $%
Z_{t}=\delta \phi _{t}$ with $\phi \in \mathcal{L}_{BMO}^{2}$, and choose $%
\left( Y_{t},\lambda \right) ,$ $t\geq 0,$ as
\begin{equation*}
Y_{t}=Y_{0}+\lambda t-\int_{0}^{t}\frac{1}{2}\frac{\delta }{1-\delta }%
|\delta \phi _{s}+\theta (V_{s})|^{2}ds-\frac{1}{2}\int_{0}^{t}|\delta \phi
_{s}|^{2}ds+\int_{0}^{t}\delta \phi _{s}^{T}dW_{s}.
\end{equation*}%
Then, $\left( Y_{t},\delta \phi _{t},\lambda \right) $ solves equation (\ref%
{EQBSDE1}), and, in turn, (\ref{power}) yields the power forward process%
\begin{equation*}
U(x,t)=\frac{x^{\delta }}{\delta }e^{Y_{0}-\int_{0}^{t}\frac{1}{2}\frac{%
\delta }{1-\delta }|\delta \phi _{s}+\theta (V_{s})|^{2}ds}\mathcal{E}\left(
\int_{0}^{\cdot }\delta \phi _{s}^{T}dW_{s}\right) _{t}
\end{equation*}%
\begin{equation*}
=e^{Y_{0}}\frac{x^{\delta }}{\delta }e^{-\int_{0}^{t}\frac{1}{2}\frac{\delta
}{1-\delta }|\phi _{s}+\theta (V_{s})|^{2}}\left( \mathcal{E}\left(
\int_{0}^{\cdot }-\phi _{s}^{T}(dW_{s}+\theta (V_{s})ds)\right) _{t}\right)
^{-\delta }
\end{equation*}%
\begin{equation*}
=e^{Y_{0}}\frac{1}{\delta }\left( \frac{x}{M_{t}}\right) ^{\delta }e^{-\frac{%
1}{2}\frac{\delta }{1-\delta }A_{t}^{\phi }},
\end{equation*}%
with $A_{t}^{\phi }=\int_{0}^{t}|\phi _{s}+\theta (V_{s})|^{2}ds$ and $M_{t}=%
\mathcal{E}\left( \int_{0}^{\cdot }-\phi _{s}^{T}(\theta
(V_{s})ds+dW_{s})\right) _{t}$. \ We may then view this process as measuring
the performance of investment strategies in relation to a \textquotedblleft
benchmark", represented by the process $M_{t}.$

For more details about the above processes and further interpretations, as
well as the specification of the associated myopic and non-myopic portfolio
components, and the corresponding wealth processes, we refer the reader to
\cite{MZ3}.

\section{Exponential case}

\label{SectionExpUtility}

We examine forward performance processes in the exponential factor-form%
\begin{equation}
U\left( x,t\right) =-e^{-\gamma x+f\left( V_{t},t\right) },
\label{exponential-general}
\end{equation}%
where $f$ is a (deterministic) function to be specified.

For exponential forward performance processes, it is more convenient for the
control policy to represent the discounted amount (and not the proportions
of the discounted wealth) invested in the individual stock accounts. Hence,
we set $\tilde{\alpha}_{t}=\tilde{\pi}_{t}X_{t}^{\pi }$. In turn, we rescale
$\tilde{\alpha}_{t}$ by the stocks' volatility, and deduce that the wealth
process solves, for $t\geq 0$,
\begin{equation}
dX_{t}^{\alpha }=\alpha _{t}^{T}(\theta (V_{t})dt+dW_{t}),
\label{wealth-expo}
\end{equation}%
with $\alpha _{t}^{T}=\tilde{\alpha}_{t}^{T}\sigma (V_{t}).$ The set of
admissible policies is $\mathcal{A},$ and we take the admissible wealth
domain to be $\mathbb{D}=\mathbb{R}$.

As in the power case, (\ref{exponential-general}) and (\ref{SPDE})\ yield
that $f$ must satisfy, for $\left( v,t\right) \in \mathbb{R}_{+}\times \left[
0,\infty \right) ,$ a semilinear PDE, given by
\begin{equation}
f_{t}+\frac{1}{2}Trace\left( \kappa \kappa ^{T}\nabla ^{2}f\right) +\eta
(v)^{T}\nabla f+G(v,\kappa ^{T}\nabla f)=0,  \label{exponential-equation}
\end{equation}%
with
\begin{equation}
G(v,z)=\frac{1}{2}\gamma ^{2}dist^{2}\left( \Pi ,\frac{z+\theta (v)}{\gamma }%
\right) -\frac{1}{2}|z+\theta (v)|^{2}+\frac{1}{2}|z|^{2},
\label{exponential-driver}
\end{equation}%
which is ill-posed with no known solutions to date. On the other hand, as in
the former case, we will construct the process $f\left( V_{t},t\right) $
itself directly from the Markovian solution of an ergodic BSDE whose driver
is of the above form (cf. (\ref{G-driver})).

\subsection{Construction via ergodic BSDE}

The results are similar to the ones derived in the previous section and are,
thus, stated without proofs.

\begin{proposition}
\label{Theorem3_EQBSDE} Assume that the market price of risk vector $\theta
\left( v\right) $ satisfies Assumption 1.ii and let the set $\Pi $ be as in (%
\ref{admissibility}). Then, the ergodic BSDE
\begin{equation}
dY_{t}=(-G(V_{t},Z_{t})+\lambda )dt+Z_{t}^{T}dW_{t},  \label{EQBSDE2}
\end{equation}%
with the driver $G(\cdot ,\cdot )$ is given by
\begin{equation}
G(V_{t},Z_{t})=\frac{1}{2}\gamma ^{2}dist^{2}\left( \Pi ,\frac{Z_{t}+\theta
(V_{t})}{\gamma }\right) -\frac{1}{2}|Z_{t}+\theta (V_{t})|^{2}+\frac{1}{2}%
|Z_{t}|^{2},  \label{G-driver}
\end{equation}%
admits a unique Markovian solution $(Y_{t},Z_{t},\lambda ),t\geq 0.$

Specifically, there exist a unique $\lambda \in \mathbb{R}$ and functions $y:%
\mathbb{R}^{d}\rightarrow \mathbb{R}$ and $z:\mathbb{R}^{d}\rightarrow
\mathbb{R}^{d}$ such that $(Y_{t},Z_{t})=(y(V_{t}),z(V_{t})).$ The function $%
y(\cdot )$ is unique up to a constant and has at most linear growth, and $%
z\left( \cdot \right) $ is bounded with $|z\left( \cdot \right) |\leq \frac{%
C_{v}}{C_{\eta }-C_{v}}$, where $C_{\eta }$ and $C_{v}$ are as in (\ref%
{dissipative}) and (\ref{driver0}), respectively.
\end{proposition}

\begin{theorem}
\label{Theorem4_ForwardUtility} Let $(Y_{t},Z_{t},\lambda
)=(y(V_{t}),z(V_{t}),\lambda ),$ $t\geq 0,$ be the unique Markovian solution
of the ergodic BSDE (\ref{EQBSDE2}). Then,

i) the process $U\left( x,t\right) ,$ given, for $\left( x,t\right) \in
\mathbb{R}\times \left[ 0,\infty \right) ,$ by
\begin{equation}
U(x,t)=-e^{-\gamma x+y(V_{t})-\lambda t}  \label{ExponentialForwardUtility}
\end{equation}%
is an exponential forward performance process with volatility%
\begin{equation*}
a\left( x,t\right) =-e^{-\gamma x+y(V_{t})-\lambda t}z(V_{t}).
\end{equation*}

ii) The optimal portfolios $\alpha _{t}^{\ast }$ and the optimal wealth
process $X_{t}^{\ast }$ are given, respectively, by
\begin{equation}
\alpha _{t}^{\ast }=Proj_{\Pi }\left( \frac{z(V_{t})+\theta (V_{t})}{\gamma }%
\right) \text{, }X_{t}^{\ast }=X_{0}+\int_{0}^{t}(\alpha _{t}^{\ast
})^{T}(\theta (V_{t})dt+dW_{t}).  \label{ExponentialOptimalStrategy}
\end{equation}
\end{theorem}

An axiomatic construction of exponential performance processes was developed
in \cite{Gordan} for semi-martingale markets. These processes have been also
used for the construction of forward indifference prices (see, among others,
\cite{MSZ}, \cite{MZ0}, \cite{MZ-Carmona} and \cite{Henderson-Hobson}) as
well as for the axiomatic construction and characterization of the so-called
maturity-independent entropy risk measures in \cite{ZZ}.

As in the power case, we may prove the following result.

\begin{proposition}
Consider the ill-posed semilinear PDE
\begin{equation}
f_{t}+\mathcal{L}f+G(v,\kappa ^{T}\nabla f)=0,  \label{expo-eqn}
\end{equation}%
$(v,t)\in \mathbb{R}^{d}\times \lbrack 0,\infty )$, with $G\left( \cdot
,\cdot \right) $ as in (\ref{exponential-driver}) (or (\ref{G-driver})) and $%
\mathcal{L}$ as in (\ref{generatorofV}). For initial condition $f(v,0)=y(v)$%
, where $y(\cdot )$ is the function appearing in the Markovian solution $%
\left( y\left( V_{t}\right) ,z\left( V_{t}\right) ,\lambda \right) $ of the
ergodic BSDE (\ref{EQBSDE2}), equation (\ref{expo-eqn}) admits a smooth
solution given by
\begin{equation*}
f(v,t)=y(v)-\lambda t.
\end{equation*}
\end{proposition}

\subsection{Representation via infinite horizon BSDE}

\label{SubSectionExpUtility}

In analogy to the results of section \ref{SubSectionPowerUtility}, we derive
an alternative representation of the exponential forward performance process
using an infinite horizon BSDE. The proof follows along similar arguments
and is, thus, omitted.

\begin{proposition}
Assume that the market price of risk vector $\theta (v)$ satisfies
Assumption 1.ii and let the set $\Pi $ be as in (\ref{admissibility}). Let $%
\rho >0$. Then, the infinite horizon BSDE
\begin{equation}
dY_{t}^{\rho }=\left( -G(V_{t},Z_{t}^{\rho })+\rho Y_{t}^{\rho }\right)
dt+\left( Z_{t}^{\rho }\right) ^{T}dW_{t},  \label{IQBSDE2}
\end{equation}%
$t\geq 0,$ with the driver $G(\cdot ,\cdot )$ as in (\ref{EQBSDE2}), admits
a unique Markovian solution. Specifically, for each $\rho >0,$ there exist
unique\textbf{\ }functions $y^{\rho }:\mathbb{R}^{d}\rightarrow \mathbb{R}$
and $z^{\rho }:\mathbb{R}^{d}\rightarrow \mathbb{R}^{d}$ such that $%
(Y_{t}^{\rho },Z_{t}^{\rho })=(y^{\rho }(V_{t}),z^{\rho }(V_{t})),$ with $%
|y^{\rho }\left( \cdot \right) |\leq \frac{K}{\rho }$ and $|z^{\rho }\left(
\cdot \right) |\leq \frac{C_{v}}{C_{\eta }-C_{v}},$ where $C_{\eta }$ as in (%
\ref{dissipative}), and $C_{v}$, $K$ given in (\ref{driver0}) and (\ref%
{driver2}), respectively.
\end{proposition}

\begin{theorem}
Let $\left( y^{\rho }\left( V_{t}\right) ,z^{\rho }\left( V_{t}\right)
\right) ,$ $t\geq 0,$ be the unique Markovian solution to the infinite
horizon BSDE (\ref{IQBSDE2}). Then,

i) the process $U^{\rho }\left( x,t\right) ,$ $\left( x,t\right) \in \mathbb{%
R}\times \left[ 0,\infty \right) ,$ given by%
\begin{equation}
U^{\rho }(x,t)=-e^{-\gamma x+y^{\rho }(V_{t})-\int_{0}^{t}\rho y^{\rho
}(V_{u})du}  \label{ExponentialForwardUtility2}
\end{equation}%
is an exponential forward performance process with volatility%
\begin{equation*}
a^{\rho }\left( x,t\right) =-e^{-\gamma x+y^{\rho }(V_{t})-\int_{0}^{t}\rho
y^{\rho }(V_{u})du}z^{\rho }(V_{t}).
\end{equation*}

ii) The optimal portfolios $\alpha _{t}^{\ast ,\rho }$ and optimal wealth
process $X_{t}^{\ast ,\rho }$ (cf. (\ref{wealth-expo})), $t\geq 0,$ are
given, respectively, as in (\ref{ExponentialOptimalStrategy}) with $z\left(
V_{t}\right) $ replaced by $z^{\rho }\left( V_{t}\right) $.
\end{theorem}

In line with Proposition \ref{PropositionPowerUtility}, we have the
following connection between the ergodic and infinite horizon
representations for exponential forward performance processes.

\begin{proposition}
\label{PropositionExpUtility} Let $U^{\rho }(x,t)$ and $U(x,t)$ be the
exponential forward performance processes (\ref{ExponentialForwardUtility})
and (\ref{ExponentialForwardUtility2}), respectively. Then, for any
reference point $v_{0}\in \mathbb{R}^{d}$, there exists a subsequence $\rho
_{n}\downarrow 0$ (depending on $v_{0}$) such that, for $\left( x,t\right)
\in \mathbb{R}\times \left[ 0,\infty \right) ,$
\begin{equation*}
\lim_{\rho _{n}\downarrow 0}\frac{U^{\rho _{n}}(x,t)e^{-y^{\rho _{n}}\left(
v_{0}\right) }}{U(x,t)}=1.
\end{equation*}%
Moreover, for $t\geq 0,$ the associated optimal portfolios satisfy
\begin{equation*}
\lim_{\rho _{n}\downarrow 0}E_{\mathbb{P}}\int_{0}^{t}\left \vert \alpha
_{u}^{\ast ,\rho _{n}}-\alpha _{u}^{\ast }\right \vert ^{2}du=0.
\end{equation*}
\end{proposition}

\subsection{Connection with the classical exponential expected utility for
long horizons}

\label{SubSectionTraditionalExpUtility}

As in section \ref{SubSectionTraditionalPowerUtility}, we discuss the
relationship between the exponential forward performance process $U\left(
x,t\right) $ and its traditional finite horizon expected utility analogue
with the latter incorporating a terminal random endowment.

To this end, let $\rho >0$ and $\left[ 0,T\right] $ be an arbitrary trading
horizon. Consider a family of maximal expected utility problems
\begin{equation}
u^{\rho }\left( x,t;T\right) =ess\sup_{\alpha \in \mathcal{A}_{{\left[ t,T%
\right] }}}E_{\mathbb{P}}\left( -e^{-\gamma (X_{T}^{\alpha }+\xi _{T})}|%
\mathcal{F}_{t},X_{t}^{\alpha }=x\right) ,  \label{ExponentialUtility}
\end{equation}%
for $\left( x,t\right) \in \mathbb{R}\times \left[ 0,T\right] $ and the
wealth process $X_{s}^{\alpha }$, $s\in \lbrack t,T]$, solving (\ref%
{wealth-expo}). The payoff $\xi _{T}$ is defined as $\xi _{T}=\frac{1}{%
\gamma }\int_{0}^{T}\rho Y_{t}^{\rho ,T}dt,$ where $Y_{t}^{\rho ,T}$ is the
solution of the finite horizon quadratic BSDE
\begin{equation*}
Y_{t}^{\rho ,T}=\int_{t}^{T}\left( G(V_{s},Z_{s}^{\rho ,T})-\rho Y_{s}^{\rho
,T}\right) ds-\int_{t}^{T}\left( Z_{s}^{\rho ,T}\right) ^{T}dW_{s},
\end{equation*}%
with the driver $G(\cdot ,\cdot )$ given in (\ref{G-driver}). The optimal
portfolios are denoted by $\alpha _{s}^{\ast ,\rho ,T}$ for $s\in \lbrack
t,T]$. We have the following convergence result.

\begin{proposition}
\label{PropositionExpUtility2} i) Let $u^{\rho }(x,t;T)$ and $U^{\rho }(x,t)$
be given in (\ref{ExponentialUtility}) and (\ref{ExponentialForwardUtility2}%
), respectively. Then, for each $\rho >0,$ and $\left( x,t\right) \in
\mathbb{R}\times \left[ 0,\infty \right) ,$
\begin{equation*}
\lim_{T\uparrow \infty }\frac{u^{\rho }(x,t;T)}{U^{\rho }(x,t)}=1,
\end{equation*}%
and the optimal portfolios satisfy, for $s\in \left[ t,T\right] ,$
\begin{equation*}
\lim_{T\uparrow \infty }E_{\mathbb{P}}\int_{t}^{s}\left \vert \alpha
_{u}^{\ast ,\rho ,T}-\alpha _{u}^{\ast ,\rho }\right \vert ^{2}du=0.
\end{equation*}

ii) Let $U\left( x,t\right) $ be the exponential forward process as in (\ref%
{ExponentialForwardUtility}). Then, for any reference point $v_{0}\in
\mathbb{R}^{d}$, there exists a subsequence $\rho _{n}\downarrow 0$
(depending on $v_{0})$ such that, for $\left( x,t\right) \in \mathbb{R}%
\times \left[ 0,\infty \right) ,$
\begin{equation*}
\lim_{\rho _{n}\downarrow 0}\lim_{T\uparrow \infty }\frac{u^{\rho
_{n}}(x,t;T)e^{-y^{\rho _{n}}}\left( v_{0}\right) }{U(x,t)}=1,
\end{equation*}%
and the optimal portfolios satisfy, for $s\in \left[ t,T\right] ,$
\begin{equation*}
\lim_{\rho _{n}\downarrow 0}\lim_{T\uparrow \infty }E_{\mathbb{P}%
}\int_{t}^{s}\left \vert \alpha _{u}^{\ast ,\rho _{n},T}-\alpha _{u}^{\ast
}\right \vert ^{2}du=0.
\end{equation*}
\end{proposition}

\section{Logarithmic case}

\label{SectionLogUtility}

We conclude with logarithmic forward performance processes in factor-form,
namely, of the form
\begin{equation}
U\left( x,t\right) =\ln x+f\left( V_{t},t\right) ,  \label{logarithmic}
\end{equation}%
for a function $f$ to be determined. The ``additive" format is more
appropriate for the logarithmic class, given the ``myopic" character of the
latter in the classical setting. Then, (\ref{logarithmic}) and (\ref{SPDE})\
yield that $f:\left( v,t\right) \in \mathbb{R}^{d}\times \left[ 0,\infty
\right) $ must satisfy the ill-posed linear equation
\begin{equation}
f_{t}+\frac{1}{2}Trace\left( \kappa \kappa ^{T}\nabla ^{2}f\right) +\eta
(v)^{T}\nabla f+\tilde{F}(v)=0,  \label{logarithmic-eqn}
\end{equation}%
with
\begin{equation*}
\tilde{F}(v)=-\frac{1}{2}dist^{2}\left \{ \Pi ,\theta (v)\right \} +\frac{1}{%
2}|\theta (v)|^{2}.
\end{equation*}%
The results that follow are similar to the ones in section 3 and, for this,
they are stated in an abbreviated manner. To this end, the associated
ergodic BSDE\ is given by
\begin{equation}
dY_{t}=(-\tilde{F}(V_{t})+\lambda )dt+Z_{t}^{T}dW_{t}  \label{log-EBSDE}
\end{equation}%
with the driver
\begin{equation*}
\tilde{F}(V_{t})=-\frac{1}{2}dist^{2}\left \{ \Pi ,\theta (V_{t})\right \} +%
\frac{1}{2}|\theta (V_{t})|^{2},
\end{equation*}%
as it is easily guessed by the form of the operator appearing in the
equation (\ref{logarithmic-eqn}) above. Working as in the proof of
Proposition \ref{Theorem1_EQBSDE} we deduce that (\ref{log-EBSDE}) has a
unique Markovian solution, say $\left( Y_{t},Z_{t},\lambda \right) =\left(
y\left( V_{t}\right) ,z\left( V_{t}\right) ,\lambda \right) ,$ for some
functions $y\left( \cdot \right) $ and $z\left( \cdot \right) $ with similar
properties to the ones therein.

We verify that the process
\begin{equation}
U(x,t):=\ln x+y\left( V_{t}\right) -\lambda t  \label{log-forward}
\end{equation}%
is a logarithmic forward performance process in factor-form. The SPDE (\ref%
{SPDE}) then yields volatility $a\left( x,t\right) =z\left( V_{t}\right) .$
Moreover, the optimal policy and the wealth it generates are given,
respectively, by $\pi _{t}^{\ast }=Proj_{\Pi }\theta (V_{t}),$ and \
\begin{equation*}
X_{t}^{\ast }=X_{0}\mathcal{E}\left( \int_{0}^{\cdot }(Proj_{\Pi }\theta
(V_{s}))^{T}(\theta (V_{s})ds+dW_{s})\right) _{t}.
\end{equation*}%
The constant $\lambda $ has the interpretation
\begin{equation*}
\lambda =\sup_{{\pi }\in \mathcal{A}}\limsup_{T\uparrow \infty }\frac{1}{T}%
E_{\mathbb{P}}\left( \ln X_{T}^{\pi }\right) .
\end{equation*}%
A by-product of this result is that the ill-posed linear PDE\ (\ref%
{logarithmic-eqn}) has a smooth solution for initial data $f\left(
v,0\right) =y\left( v\right) ,$ given by $f\left( v,t\right) =y\left(
v\right) -\lambda t.$

There is also a connection with infinite horizon BSDE. Indeed, we easily
deduce that the infinite horizon BSDE
\begin{equation}
dY_{t}^{\rho }=\left( -\tilde{F}(V_{t})+\rho Y_{t}^{\rho }\right) dt+\left(
Z_{t}^{\rho }\right) ^{T}dW_{t},  \label{log-BSDE}
\end{equation}%
has a unique Markovian solution $\left( y^{\rho }\left( V_{t}\right)
,z^{\rho }\left( V\right) \right) ,$ and, in turn, the process $U^{\rho
}(x,t)$, $(x,t)\in \mathbb{R}_{+}\times \lbrack 0,\infty )$, defined as
\begin{equation}
U^{\rho }(x,t):=\ln x+y^{\rho }\left( V_{t}\right) -\int_{0}^{t}\rho y^{\rho
}\left( V_{s}\right) ds,  \label{LogForwardUtility2}
\end{equation}%
is a path-dependent logarithmic forward performance process.

The process $U(x,t)$ and $U^{\rho }(x,t)$ in (\ref{log-forward}) and (\ref%
{LogForwardUtility2}) are connected in a similar way as their power
analogues in Proposition \ref{PropositionPowerUtility}. Namely, for an
arbitrary reference point $v_{0}\in \mathbb{R}^{d}$, there exists a
subsequence $\rho _{n}\downarrow 0$ (depending on $v_{0})$ such that, for $%
(x,t)\in \mathbb{R}_{+}\times \lbrack 0,\infty )$,
\begin{equation*}
\lim_{\rho _{n}\downarrow 0}\left( U^{\rho _{n}}(x,t)-y^{\rho _{n}}\left(
v_{0}\right) -U(x,t)\right) =0.
\end{equation*}

Finally, in order to connect $U(x,t)$ and $U^{\rho }(x,t)$ with their
classical counterparts, we introduce the logarithmic expected utility
problem
\begin{equation}
u^{\rho }(x,t;T)=ess\sup_{\pi \in \mathcal{A}_{\left[ t,T\right] }}{E}_{%
\mathbb{P}}\left( \ln X_{T}^{\pi }+\xi _{T}|\mathcal{F}_{t},X_{t}^{\pi
}=x\right) ,  \label{LogUtility}
\end{equation}%
where $\xi _{T}=-\int_{0}^{T}\rho Y_{u}^{\rho ,T}du$ and $Y_{t}^{\rho ,T},$ $%
t\in \left[ 0,T\right] ,$ is the unique solution of the BSDE on $[0,T],$
\begin{equation*}
Y_{t}^{\rho ,T}=\int_{t}^{T}\left( \tilde{F}(V_{u})-\rho Y_{u}^{\rho
,T}\right) du-\int_{t}^{T}\left( Z_{u}^{\rho ,T}\right) ^{T}dW_{u}.
\end{equation*}%
Using similar arguments to the ones in Proposition \ref%
{PropositionPowerUtility2}, we deduce that for any reference point $v_{0}\in
\mathbb{R}^{d}$, there exists a subsequence $\rho _{n}\downarrow 0$
(depending on $v_{0})$ such that, for $(x,t)\in \mathbb{R}_{+}\times \lbrack
0,\infty )$,
\begin{equation*}
\lim_{\rho _{n}\downarrow 0}\lim_{T\uparrow \infty }\left( u^{\rho
_{n}}(x,t;T)-y^{\rho _{n}}\left( v_{0}\right) -U(x,t)\right) =0.
\end{equation*}

\label{SectionRemark}

\appendix

\section{Appendix: Solving ergodic and infinite horizon BSDE}

\label{SectionEQBSDE}

We present background results for Markovian solutions of the ergodic BSDE (%
\ref{EQBSDE1}) and (\ref{EQBSDE2}). We also obtain existence and uniqueness
of bounded Markovian solutions to the infinite horizon BSDE (\ref{IQBSDE1})
and (\ref{IQBSDE2}) as intermediate steps in the proofs of Propositions \ref%
{Theorem1_EQBSDE} and \ref{Theorem3_EQBSDE}. The equations (\ref{log-EBSDE})
and (\ref{log-BSDE}) appearing in the logarithmic case are degenerate
versions on (\ref{EQBSDE1}) and (\ref{IQBSDE1}), so they are not discussed.

We start with the key observation that, using Assumption 1.\textit{ii} on
the market price of risk process as well as the definition of the admissible
set $\mathcal{A}$ and the Lipschitz continuity of the distance function $%
dist\left( \Pi ,\cdot \right) $, we deduce that the drivers $H=F$, $G$
appearing in (\ref{driver-power}) and (\ref{G-driver}) satisfy
\begin{equation}
|H(v,z)-H(\bar{v},z)|\leq C_{v}(1+|z|)|v-\bar{v}|,  \label{driver0}
\end{equation}%
\begin{equation}
|H(v,z)-H(v,\bar{z})|\leq C_{z}(1+|z|+|\bar{z}|)|z-\bar{z}|,  \label{driver1}
\end{equation}%
and
\begin{equation}
|H(v,0)|\leq K,  \label{driver2}
\end{equation}%
for any $v,\bar{v},z,\bar{z}\in \mathbb{R}^{d}$, and $C_{v}$, $C_{z}$, $K>0$
being positive constants.

The main ideas for establishing existence and uniqueness of solutions come
from Theorem 3.3 in \cite{Briand0}, Theorem 3.3 in \cite{Briand}, Theorem
4.4 in \cite{HU2}, and Theorem 2.3 in \cite{Kobylanski}. To this end, we
first define the truncation function $q:$ $\mathbb{R}^{d}\rightarrow \mathbb{%
R}^{d},$
\begin{equation}
q(z):=\frac{\min \left( |z|,C_{v}/(C_{\eta }-C_{v})\right) }{|z|}z\mathbf{1}%
_{\{z\neq 0\}},  \label{truncation}
\end{equation}%
and consider the truncated ergodic BSDE,
\begin{equation}
d{Y}_{t}=\left( -{H}(V_{t},q({Z}_{t}))+\lambda \right) dt+{Z}_{t}^{T}dW_{t},
\label{EQBSDE3}
\end{equation}%
$t\geq 0,$ where $q$ is as in (\ref{truncation}), and the driver ${H}(\cdot
,\cdot )$ satisfies conditions (\ref{driver0})-(\ref{driver2}). We easily
obtain the Lipschitz continuity conditions
\begin{equation}
|{H}(v,q(z))-{H}(\bar{v},q(z)|\leq \frac{C_{\eta }C_{v}}{C_{\eta }-C_{v}}|v-%
\bar{v}|,  \label{driver3}
\end{equation}%
and%
\begin{equation}
|H(v,q(z))-{H}(v,q(\bar{z})|\leq C_{z}\frac{C_{\eta }+C_{v}}{C_{\eta }-C_{v}}%
|z-\bar{z}|.  \label{driver4}
\end{equation}%
If, therefore, we can show that the BSDE (\ref{EQBSDE3}) admits a Markovian
solution denoted, say by $(Y_{t},Z_{t},{\lambda })$ with $|{Z}_{t}|\leq
\frac{C_{v}}{C_{\eta -C_{v}}}$, $t\geq 0$, then $q({Z}_{t})={Z}_{t}$, $t\geq
0$. In turn, this process $({Y}_{t},{Z}_{t},{\lambda })$ would also solve
the ergodic BSDE (\ref{EQBSDE1}) in Proposition \ref{Theorem1_EQBSDE} and (%
\ref{EQBSDE2}) in Proposition \ref{Theorem3_EQBSDE}, respectively.

We first establish existence of Markovian solutions of (\ref{EQBSDE3}). For
this, we adapt the perturbation technique and the Girsanov's transformation
used in Section 4 of \cite{HU2} in an infinite dimensional setting. To this
end, let $n>0,$ and consider the discounted BSDE with a small discount
factor, say $\rho >0$, on the finite horizon $[0,n]$,
\begin{equation}
Y_{t}^{\rho ,v,n}=\int_{t}^{n}\left( H(V_{s}^{v},q(Z_{s}^{\rho ,v,n}))-\rho
Y_{s}^{\rho ,v,n}\right) ds-\int_{t}^{n}\left( Z_{s}^{\rho ,v,n}\right)
^{T}dW_{s},  \label{QBSDE}
\end{equation}%
where we use the superscript $v$ to emphasize the initial dependence of the
stochastic factor process on its initial data $V_{0}^{v}=v$.

From Section 3.1 of \cite{Briand}, we deduce that BSDE (\ref{QBSDE}) admits
a unique solution $(Y_{t}^{\rho ,v,n},Z_{t}^{\rho ,v,n})\in \mathcal{L}%
_{[0,n]}^{2}$ with $|Y_{t}^{\rho ,v,n}|\leq \frac{K}{\rho }$, $0\leq t\leq
n, $ where
\begin{equation*}
\mathcal{L}^{2}\left[ 0,n\right] =\left \{ (Y_{t})_{t\in \lbrack 0,n]}:Y\
\text{is}\  \mathbb{F}\text{-progressively\ measurable}\right. \
\end{equation*}%
\begin{equation*}
\left. \text{and \ }E_{\mathbb{P}}(\int_{0}^{n}|Y_{t}|^{2}dt)<\infty \right
\} .
\end{equation*}%
On the other hand, parameterizing (\ref{QBSDE}) by the auxiliary horizon $n$%
, we obtain (cf. section 3.1 of \cite{Briand}) that there exists a process $%
Y_{t}^{\rho ,v},$ $t\geq 0,$ such that
\begin{equation}
\lim_{n\uparrow \infty }Y_{t}^{\rho ,v,n}=Y_{t}^{\rho ,v},
\label{convergence0}
\end{equation}%
for a.e. $(t,\omega )\in \lbrack 0,\infty )\times \Omega $, and moreover,
that for each $\rho >0,$ both $\{Y_{t}^{\rho ,v,n}\}$ and $\{Z_{t}^{\rho
,v,n}\}$ are Cauchy sequences in $\mathcal{L}_{\rho }^{2}\left[ 0,\infty %
\right] $, where
\begin{equation}
\mathcal{L}_{\rho }^{2}[0,\infty )=\left \{ (Y_{t})_{t\in \lbrack 0,\infty
)}:Y\  \text{is}\  \mathbb{F}\text{-progressively\ measurable}\right.
\label{L-2-rho}
\end{equation}%
\begin{equation*}
\left. \text{and\  \ }E_{\mathbb{P}}(\int_{0}^{\infty }e^{-2\rho
t}|Y_{t}|^{2}dt)<\infty \right \} .
\end{equation*}%
Therefore, there exist limiting processes $(Y_{t}^{\rho ,v},Z_{t}^{\rho
,v}), $ $t\geq 0,$ belonging to $\mathcal{L}_{\rho }^{2}[0,\infty ),$ such
that
\begin{equation}
\lim_{n\uparrow \infty }(Y_{t}^{\rho ,v,n},Z_{t}^{\rho ,v,n})=(Y_{t}^{\rho
,v},Z_{t}^{\rho ,v})  \label{convergence1}
\end{equation}%
in $\mathcal{L}_{\rho }^{2}[0,\infty )$ with $|Y_{t}^{\rho ,v}|\leq \frac{K}{%
\rho }.$ It is, then, easy to show that the process $(Y_{t}^{\rho
,v},Z_{t}^{\rho ,v})$, $t\geq 0$, is a solution to the {infinite horizon }%
BSDE
\begin{equation}
dY_{t}^{\rho ,v}=\left( -H(V_{t}^{v},q(Z_{t}^{\rho ,v}))+\rho Y_{t}^{\rho
,v}\right) dt+\left( Z_{t}^{\rho ,v}\right) ^{T}dW_{t}.  \label{QBSDE2}
\end{equation}%
Moreover, we recall that the solution is Markovian in the sense that there
exist functions, say $y^{\rho }(\cdot )$ and $z^{\rho }(\cdot )$, such that
\begin{equation*}
\left( Y_{t}^{\rho ,v},Z_{t}^{\rho ,v}\right) =\left( y^{\rho
}(V_{t}^{v}),z^{\rho }(V_{t}^{v})\right) .
\end{equation*}

{Next, using the Girsanov's transformation and adapting the argument in
Lemma 4.3 in \cite{HU2}, we claim that the Lipschitz continuity property}%
\textbf{\ }%
\begin{equation}
|y^{\rho }(V_{t}^{v})-y^{\rho }(V_{t}^{\bar{v}})|\leq \frac{C_{v}}{C_{\eta
}-C_{v}}|V_{t}^{v}-V_{t}^{\bar{v}}|  \label{Lip}
\end{equation}%
holds, {for any }$v,\bar{v}\in \mathbb{R}^{d}${, with the constants }$C_{v}$
and $C_{\eta }$ as in (\ref{driver0}) and (\ref{dissipative}), respectively.

Indeed, define, for $t\geq 0,$
\begin{equation*}
\Delta Y_{t}:=Y_{t}^{\rho ,v}-Y_{t}^{\rho ,\bar{v}}\text{ \  \  \ and \  \  \ }%
\Delta Z_{t}:=Z_{t}^{\rho ,v}-Z_{t}^{\rho ,\bar{v}}.
\end{equation*}%
Then,
\begin{equation*}
d\left( \Delta Y_{t}\right) =-\left( H(V_{t}^{v},q(Z_{t}^{\rho
,v}))-H(V_{t}^{\bar{v}},q(Z_{t}^{\rho ,\bar{v}}))\right) dt+\rho \Delta
Y_{t}dt+\left( \Delta Z_{t}\right) ^{T}dW_{t}
\end{equation*}%
\begin{equation*}
=-\Delta H_{t}dt+\rho \Delta Y_{t}dt+\left( \Delta Z_{t}\right) ^{T}\left(
dW_{t}-m_{t}dt\right) ,
\end{equation*}%
where $\Delta H_{t}:=H(V_{t}^{v},q(Z_{t}^{\rho ,v}))-H(V_{t}^{\bar{v}%
},q(Z_{t}^{\rho ,v}))$ and%
\begin{equation*}
m_{t}:=\frac{H(V_{t}^{\bar{v}},q(Z_{t}^{\rho ,v}))-H(V_{t}^{\bar{v}%
},q(Z_{t}^{\rho ,\bar{v}}))}{|\Delta Z_{t}|^{2}}\Delta Z_{t}\mathbf{1}_{\{
\Delta Z_{t}\neq 0\}}.
\end{equation*}%
The process $m_{t}$ is bounded, as it follows from (\ref{driver4}).
Therefore, we can define the process $\bar{W}_{t}:=W_{t}-\int_{0}^{t}m_{u}du$%
, $t\geq 0$, which is a Brownian motion under some measure $\mathbb{Q}^{m}$
equivalent to $\mathbb{P}$. Hence, for $0\leq t\leq s<\infty $, taking
conditional expectation on $\mathcal{F}_{t}$ under $\mathbb{Q}^{m}$ yields
\begin{equation*}
\Delta Y_{t}=\frac{\beta _{s}}{\beta _{t}}E_{\mathbb{Q}^{m}}\left( \Delta
Y_{s}|\mathcal{F}_{t}\right) +E_{\mathbb{Q}^{m}}\left( \left. \int_{t}^{s}%
\frac{\beta _{u}}{\beta _{t}}\left( \Delta H_{u}\right) du\right \vert
\mathcal{F}_{t}\right) ,
\end{equation*}%
where $\beta _{t}=e^{-\rho t}$. Note, however, that the first expectation
above is bounded by $2K/\rho $, and thus, it converges to zero as $s\uparrow
\infty $. Moreover, by (\ref{driver3}), the second expectation is bounded by%
\begin{equation*}
E_{\mathbb{Q}^{m}}\left( \left. \int_{t}^{s}\frac{\beta _{u}}{\beta _{t}}%
\left( \Delta H_{u}\right) du\right \vert \mathcal{F}_{t}\right) \leq \frac{%
C_{\eta }C_{v}}{C_{\eta }-C_{v}}E_{\mathbb{Q}^{m}}\left( \left.
\int_{t}^{s}e^{-\rho (u-t)}|V_{u}^{v}-V_{u}^{\bar{v}}|du\right \vert
\mathcal{F}_{t}\right)
\end{equation*}%
\begin{equation*}
\leq \  \frac{C_{\eta }C_{v}}{C_{\eta }-C_{v}}\frac{e^{{\rho }t}\left(
e^{-(\rho +C_{\eta })t}-e^{-(\rho +C_{\eta })s}\right) }{\rho +C_{\eta }}|v-%
\bar{v}|,
\end{equation*}%
where we used the exponential ergodicity condition (\ref%
{exponentialcondition}). Then, inequality (\ref{Lip}) follows by letting $%
s\uparrow \infty $.

Next, assume that $y^{\rho }(\cdot )\in C^{2}(\mathbb{R}^{d})$. Applying It%
\^{o}'s formula to $y^{\rho }(V_{t}^{v})$ yields
\begin{equation}
dy^{\rho }(V_{t}^{v})=\mathcal{L}y^{\rho }(V_{t}^{v})dt+\left( \kappa
^{T}\nabla y^{\rho }(V_{t}^{v})\right) ^{T}dW_{t},  \label{ITO}
\end{equation}%
where $\mathcal{L}$ is as in (\ref{generatorofV}). In turn, from (\ref%
{QBSDE2}) and (\ref{ITO}) we deduce that
\begin{equation}
\kappa ^{T}\nabla y^{\rho }(V_{t}^{v})=Z_{t}^{\rho ,v},  \label{relationZ}
\end{equation}%
and (with a slight abuse of notation) that
\begin{equation}
\rho y^{\rho }(v)=\mathcal{L}y^{\rho }(v)+H\left( v,q\left( \kappa
^{T}\nabla {y}^{\rho }(v)\right) \right) ,  \label{ellipticPDE}
\end{equation}%
for $v\in \mathbb{R}^{d}.$ The above equation (\ref{ellipticPDE}) is a
standard semilinear elliptic PDE, and classical PDE results yield that it
admits a unique bounded solution $y^{\rho }(\cdot )\in C^{2}(\mathbb{R}^{d})$%
, with $|y^{\rho }(v)|\leq \frac{K}{\rho }$. In addition, recall that (\ref%
{Lip}) yields $|\nabla y^{\rho }(v)|\leq \frac{C_{v}}{C_{\eta }-C_{v}},$ and
thus, using (\ref{relationZ}) and Assumption 2 on the matrix $\kappa $, we
obtain that, for $t\geq 0$,
\begin{equation}
|Z_{t}^{\rho ,v}|\leq \frac{C_{v}}{C_{\eta }-C_{v}}.  \label{estimateZ}
\end{equation}

Next, we fix a reference point, say $v_{0}\in \mathbb{R}^{d}$. Define the
process $\bar{Y}_{t}^{\rho ,v}:=Y_{t}^{^{\rho ,v}}-Y_{0}^{\rho ,v_{0}},$ and
consider the perturbed version of the {infinite horizon} BSDE (\ref{QBSDE2}%
), namely,
\begin{equation*}
\bar{Y}_{t}^{\rho ,v}=\bar{Y}_{s}^{\rho ,v}+\int_{t}^{s}\left(
H(V_{u}^{v},q(Z_{u}^{\rho ,v}))-\rho \bar{Y}_{u}^{^{\rho ,v}}-\rho
Y_{0}^{\rho ,v_{0}}\right) du-\int_{t}^{s}\left( Z_{u}^{\rho ,v}\right)
^{T}dW_{u},
\end{equation*}%
for $0\leq t\leq s<\infty $. Then $\bar{Y}_{t}^{\rho ,v}=\bar{y}^{\rho
}(V_{t}^{v})$ with $\bar{y}^{\rho }(\cdot )$ $=y^{\rho }(\cdot )-y^{\rho
}(v_{0})$.

Since, on the other hand, $y^{\rho }(\cdot )$ is Lipschitz continuous,
uniformly in $\rho ,$ we deduce that $\left \vert \bar{y}^{\rho
}(v)\right
\vert \leq \frac{C_{v}}{C_{\eta }-C_{v}}\left \vert
v-v_{0}\right \vert .$ Moreover, $|\rho y^{\rho }(v)|\leq K.$\textbf{\emph{\
}}Hence, there exists a sequence $\rho _{0n}\downarrow 0$ such that
\begin{equation*}
\lim_{\rho _{0n}\downarrow 0}\rho _{0n}y^{\rho _{0n}}(v_{0})=\lambda ,
\end{equation*}%
for some constant $\lambda $.

Next, we take a dense subset, say $S=\{v_{1},\cdots ,v_{n},\cdots \} \in
\mathbb{R}^{d}$. Since $\bar{y}^{\rho _{0n}}(v_{1})$ is bounded, there
exists a subsequence of $\{ \rho _{0n}\}$, denoted as $\{ \rho _{1n}\}$,
such that
\begin{equation*}
\lim_{\rho _{1n}\downarrow 0}\bar{y}^{\rho _{1n}}(v_{1})=y(v_{1}),
\end{equation*}%
for some $y(v_{1})$. Proceeding this way, we obtain a sequence $\{ \rho
_{0n}\} \supset \{ \rho _{1n}\} \supset \cdots $. Taking its diagonal
sequence $\{ \rho _{nn}\}$, denoted as $\{ \rho _{n}\}$, we deduce that, for
$v\in S$,
\begin{equation}
\lim_{\rho _{n}\downarrow 0}\rho _{n}y^{\rho _{n}}(v_{0})=\lambda \text{ \  \
and \ }\lim_{\rho _{n}\downarrow 0}\bar{y}^{\rho _{n}}(v)=y\left( v\right) .%
\text{ \  \ }  \label{**}
\end{equation}

Moreover, since the function $\bar{y}^{\rho }(\cdot )$ is Lipschitz
continuous uniformly in $\rho $, the limit $y(\cdot )$ can be extended to a
Lipschitz continuous function defined for all $v\in \mathbb{R}^{d}$, \
\begin{equation}
\lim_{\rho _{n}\downarrow 0}\bar{y}^{\rho _{n}}(v)=y\left( v\right) .
\label{*}
\end{equation}

Thus, we have $\lim_{\rho _{n}\downarrow 0}\bar{Y}_{t}^{\rho _{n},v}=y\left(
V_{t}^{v}\right) $ \  \ and \ $\lim_{\rho _{n}\downarrow 0}\left( \rho _{n}%
\bar{Y}_{t}^{\rho _{n}}\right) =0.$

Next, define the process ${Y}_{t}^{v}=y(V_{t}^{v}),$ $t\geq 0.$ It is then
standard to show that there exists ${Z}_{u}^{v}=z(V_{u}^{v})$, $u\in \lbrack
t,s]$, in $\mathcal{L}^{2}[t,s]$ such that $\lim_{\rho _{n}\downarrow
0}Z^{\rho _{n},v}={Z}^{v}$ in $\mathcal{L}^{2}[t,s]$, and moreover, that the
triplet $({Y}_{t}^{v},{Z}_{t}^{v},{\lambda })=(y(V_{t}^{v}),z(V_{t}^{v}),%
\lambda )$ is a solution to the truncated ergodic BSDE (\ref{EQBSDE3}).

Finally, using the latter limit and the fact that $|Z_{t}^{\rho ,v}|\leq
C_{v}/(C_{\eta }-C_{v}),$ as it follows from (\ref{estimateZ}), we obtain
that $|{Z}_{t}^{v}|\leq C_{v}/(C_{\eta }-C_{v})$. Therefore, $q({Z}_{t}^{v})=%
{Z}_{t}^{v}$, $t\geq 0$, and in turn, the triplet $({Y}^{v},{Z}^{v},{\lambda
})$ is also a solution to the ergodic BSDEs (\ref{EQBSDE1}) and (\ref%
{EQBSDE2}) in Propositions \ref{Theorem1_EQBSDE} and \ref{Theorem3_EQBSDE},
respectively.

From the above arguments, it follows, as a by-product, the existence of
Markovian solutions to the infinite horizon BSDEs (\ref{IQBSDE1}) and (\ref%
{IQBSDE2}), respectively.

It remains to show the uniqueness of Markovian solutions to the ergodic BSDE
(\ref{EQBSDE1}) and (\ref{EQBSDE2}). Indeed, since $Z_{t}$, $t\geq 0$, is
bounded by $C_{v}/(C_{\eta }-C_{v})$ for both equations (\ref{EQBSDE1}) and (%
\ref{EQBSDE2}), the uniqueness can be proved along similar arguments used in
Theorem 4.6 in \cite{HU2} and Theorem 3.11 in \cite{HU1}.

The uniqueness of the Markovian solutions to the infinite horizon BSDE (\ref%
{IQBSDE1}) and (\ref{IQBSDE2}) follows easily from Section 3.1 in \cite%
{Briand} and Theorem 3.3 in \cite{Briand0}.


\newif \ifabfull \abfulltrue
\providecommand{\bysame}{\leavevmode \hbox to3em{\hrulefill}\thinspace} %
\providecommand{\MR}{\relax \ifhmode \unskip \space \fi MR }
\providecommand{\MRhref}[2]{
\href{http://www.ams.org/mathscinet-getitem?mr=#1}{#2} } \providecommand{%
\href}[2]{#2}

\end{document}